\documentclass[a4paper, UKenglish, cleveref, autoref, thm-restate]{lipics-v2021}

\usepackage{cite}
\usepackage{graphicx}
\usepackage{textcomp}
\usepackage{xcolor}

\usepackage{xspace}
\usepackage{bbm}
\usepackage{comment}
\usepackage{amsmath}
\usepackage{booktabs}
\usepackage{amsfonts}
\usepackage{amssymb}
\usepackage{amsthm}
\usepackage{thmtools,thm-restate}
\usepackage[textsize=small]{todonotes}
\usepackage{tikz}
\usetikzlibrary{arrows,arrows.meta,automata,positioning} 
\usepackage{stmaryrd}
\usepackage{mathdots}
\usepackage{algorithm}
\usepackage[noend]{algpseudocode}
\usepackage{thmtools}
\usepackage{lineno}
\usepackage{enumerate}
\usepackage{apxproof}
\usepackage{hyperref}
\usepackage{url}

\title{Reachability in 3-VASS is Elementary}

\author{Wojciech Czerwi\'nski}
{University of Warsaw}
{wczerwin@mimuw.edu.pl}
{https://orcid.org/0000-0002-6169-868X}
{Supported by the ERC grant INFSYS, agreement no. 950398.}

\author{Isma{\"e}l Jecker}
{FEMTO-ST, CNRS, Univ. Franche-Comté, France}
{ismael.jecker@gmail.com}
{}
{}

\author{Sławomir Lasota}
{University of Warsaw}
{s.lasota@uw.edu.pl}
{https://orcid.org/0000-0001-8674-4470}
{Supported by the ERC grant INFSYS, agreement no. 950398 and by the NCN grant 2021/41/B/ST6/00535.}

\author{Łukasz Orlikowski}
{University of Warsaw}
{l.orlikowski@mimuw.edu.pl}
{https://orcid.org/0009-0001-4727-2068}
{Supported by the ERC grant INFSYS, agreement no. 950398.}

\authorrunning{W. Czerwi{\'n}ski, I. Jecker, S. Lasota, and {\L}. Orlikowski} 
\Copyright{Wojciech Czerwi{\'n}ski, Isma{\"e}l Jecker, S{\l}awomir Lasota, and {\L}ukasz Orlikowski}

\ccsdesc[500]{Theory of computation~Parallel computing models}
\keywords{vector addition systems, Petri nets, reachability problem, dimension three, doubly exponential space, length of shortest path} 

\newcommand{\ignore}[1]{}

\newcommand{\Sandwich}{Polynomially approximable\xspace}
\newcommand{\sandwich}{polynomially approximable\xspace}
\newcommand{\parsandwich}[1]{#1-approximable\xspace}

\newcommand{\C}{\mathcal{C}}
\newcommand{\Oo}{\mathcal{O}}

\newcommand{\F}{\mathcal{F}}

\newcommand{\N}{\mathbb{N}}
\newcommand{\Z}{\mathbb{Z}}

\newcommand{\Q}{\mathbb{Q}}
\newcommand{\U}{\mathcal{U}}

\newcommand{\Qpos}{\Q_{>0}}

\newcommand{\set}[1]{\{#1\}}
\newcommand{\setof}[2]{\set{#1 \mid #2}}
\newcommand{\card}[1]{\left|#1\right|}
\newcommand{\prettyexists}[2]{\exists_{#1} \, #2}
\newcommand{\prettyforall}[2]{\forall_{#1} \, #2}

\newcommand{\fin}{\textup{fin}}

\newcommand{\reach}{\textsc{Reach}}

\newcommand{\trans}[1]{\stackrel{#1}{\longrightarrow}}
\newcommand{\tran}{\trans{*}}

\newcommand{\slps}{\textsc{slps}\xspace}

\newcommand{\norm}{\textsc{norm}}

\newcommand{\size}{\textsc{size}}
\newcommand{\spn}{\textsc{span}}

\newcommand{\eff}{\textsc{eff}}
\newcommand{\drop}{\textsc{drop}}

\newcommand{\poly}{\textup{poly}}

\newcommand{\eps}{\varepsilon}

\newcommand{\nl}{\textsc{NL}\xspace}

\newcommand{\np}{\textsc{NP}\xspace}
\newcommand{\pspace}{\textsc{PSpace}\xspace}

\newcommand{\expspace}{\textsc{ExpSpace}\xspace}
\newcommand{\twoexptime}{\textsc{2-ExpTime}\xspace}
\newcommand{\twoexpspace}{\textsc{2-ExpSpace}\xspace}
\newcommand{\tower}{\textsc{Tower}\xspace}
\newcommand{\ackermann}{\textsc{Ackermann}\xspace}

\renewcommand{\vec}[1]{\overrightarrow{#1}}

\newcommand{\subseteqfin}{\subseteq_{\fin}}
\newcommand{\LPS}{\lps}
\newcommand{\SLPS}{\slps}
\newcommand{\kanapka}[2]{(#1, #2)-approximately semi-linear}
\newcommand{\lb}{length-bounded\xspace}
\newcommand{\plb}{polynomially length-bounded\xspace}

\newcommand{\kbound}[2]{{#1}^{2^{2^{#2}}}}

\newcommand{\Lin}[1]{\textsc{Lin}(#1)}
\newcommand{\Len}[3]{\textsc{Len}(#1, #2, #3)}
\newcommand{\OO}{{\cal O}}

\newcommand{\mytrim}[2]{(#1, #2)-trim}
\newcommand{\vass}{{\sc vass}\xspace}
\newcommand{\zvass}{$\Z$-{\sc vass}\xspace}
\newcommand{\lps}{\textsc{lps}\xspace}
\newcommand{\dlps}{2-\lps}
\newcommand{\dslps}{2-\slps}
\newcommand{\dvass}{\parvass 2}
\newcommand{\tvass}{\parvass 3}
\newcommand{\geomvass}{geometrically 2-dimensional \tvass}
\newcommand{\Geomvass}{Geometrically 2-dimensional \tvass}
\newcommand{\tzvass}{\parzvass 3}
\newcommand{\parvass}[1]{{$#1$-\vass}\xspace}
\newcommand{\parzvass}[1]{{$#1$-$\Z$-\vass}\xspace}
\newcommand{\ltvass}{(V_1) u_1 (V_2) u_2 \ldots u_{\ell-1} (V_\ell)}
\newcommand{\ktvass}{(V_1) u_1 (V_2) u_2 \ldots u_{k-1} (V_k)}
\newcommand{\ktvasstwo}{(V_2) u_2 \ldots u_{k-1} (V_k)}
\newcommand{\ktvassmod}{(\essdvass V_i) \essdvass u_{i, b} (V_2) u_2 \ldots u_{k-1} (V_k)}
\newcommand{\ktvassmodmod}{(\essdvass V) \essdvass u_{b} (V_2) u_2 \ldots u_{k-1} (V_k)}
\renewcommand{\C}{\mathcal{C}}
\newcommand{\para}[1]{\vspace{-3mm}\subparagraph*{\bf #1.}}
\newcommand{\setfromto}[2]{[#1, #2]}
\newcommand{\setto}[1]{\setfromto 1 {#1}}
\newcommand{\Wlog}{W.l.o.g.~}
\newcommand{\mywlog}{w.l.o.g.~}
\newcommand{\rev}[1]{{#1}^{\text{rev}}}

\newcommand{\cone}[1]{\textsc{Cone}(#1)}
\newcommand{\seqcone}[1]{\textsc{SeqCone}(#1)}
\newcommand{\Qnonneg}{\Q_{\geq 0}}
\newcommand{\Npos}{\N_{>0}}

\newcommand{\essdvass}[1]{\overline{#1}}
\newcommand{\innprod}[2]{#1 \diamond #2} 
\newcommand{\pair}[2]{#1_{#2}}

\newcommand{\absv}[1]{|#1|}
\newcommand{\rhs}{\textsc{rhs}}

\newcommand{\probdef}[3]{
\begin{itemize} 
    \item[] \textsc{#1}
    \item[] \textit{Input:} \quad\, #2.
    \item[] \textit{Question:} #3?
\end{itemize}
}

\nolinenumbers

\begin{document}

\maketitle

\begin{abstract}
The reachability problem in 3-dimensional vector addition systems with states (3-VASS) is known to be PSpace-hard, and to belong to Tower. We significantly narrow down the complexity gap by proving the problem to be solvable in doubly-exponential space.
The result follows from a new upper bound on the length of the shortest path:
if there is a path between two configurations of a 3-VASS then there is also one of at most triply-exponential length.
We show it by introducing a novel technique of approximating the reachability sets of 2-VASS by small semi-linear sets.
\end{abstract}

\maketitle


\section{Introduction}\label{sec:intro}

Petri nets are an established model of concurrent systems with extensive applications in various
areas of theoretical computer science.
For most algorithmic questions, the model is equivalent to
\emph{vector addition systems} with states (\vass in short).
A $d$-dimensional \vass (\parvass d in short) is a finite automaton equipped additionally
with a finite number $d$ of nonnegative integer counters that are updated by
transitions, under the proviso that the counter values can not drop below zero.
Importantly, \vass have no capability to zero-test counters, and hence the model is not Turing-complete. 

One of the central algorithmic problems for VASS is the \emph{reachability problem} which asks,
if in a given \vass there is a path (a sequence of executions of transitions) 
from a given source configuration (consisting of a 
state together with counter values) to a given target configuration:

\probdef{\vass reachability problem}
{\vass $V$, source and target configurations $s, t$}
{Is there a path from $s$ to $t$}

\noindent
Already in 1976, the problem was shown to be \expspace-hard by Lipton~\cite{Lipton76}, and
few years later decidability was shown by Mayr~\cite{Mayr81}.
Later improvements~\cite{DBLP:conf/stoc/Kosaraju82,DBLP:journals/tcs/Lambert92} simplified the construction, 
but no complexity upper bound was given until Leroux and Schmitz showed that the problem 
can be solved in Ackermannian complexity~\cite{LerouxS15,LS19}.
At the same time the lower bound was lifted to \tower-hardness~\cite{CzerwinskiLLLM21}, and soon
afterwards, in 2021, improved to \ackermann-hardness~\cite{DBLP:conf/focs/Leroux21,DBLP:conf/focs/CzerwinskiO21}.

Although the complexity of the reachability problem is now settled to be \ackermann-complete, 
various complexity questions remain still widely open.
One of them is the complexity of the reachability problem parametrised by dimension, namely in
$d$-dimensional \vass (\parvass d) for fixed $d \in \N$. 
Although the question has been investigated for a few decades, 
exact bounds are only known for dimensions 1 and 2
(both for unary or binary representations of numbers in counter updates).
In the case of binary encoding, the reachability problem is \np-complete
for \parvass 1~\cite{DBLP:conf/concur/HaaseKOW09} and \pspace-complete 
for \dvass~\cite{BlondinFGHM15}, and in case of
unary encoding the problem is \nl-complete both for \parvass 1 (folklore) and for 
\dvass~\cite{DBLP:conf/lics/EnglertLT16}.
For higher dimension almost nothing is  known.

All the upper complexity bounds for dimension 1 or 2 were obtained by estimating
the length of the shortest path, or of its representation.
For unary \parvass 1, it is a folklore that if there is a path between two configurations then there is 
also one of polynomial length (see~\cite{DBLP:journals/lmcs/ChistikovCHPW19} for more demanding 
quadratic upper bound), which implies \nl-completeness.
For binary \parvass 1, a polynomial-size representation of the shortest path was provided
by~\cite{DBLP:conf/concur/HaaseKOW09}.
Concerning binary \dvass,
already in 1979 Hopcroft and Pansiot showed that the reachability sets of \dvass
are effectively semi-linear, and therefore the reachability problem is 
decidable~\cite{DBLP:journals/tcs/HopcroftP79}.
Subsequently, \twoexptime upper complexity bound for binary \dvass
was established by~\cite{DBLP:journals/tcs/HowellRHY86}.
In~\cite{DBLP:conf/concur/LerouxS04} Leroux and Sutre showed that even the reachability relation 
is semi-linear, and that the relation
is flattable, namely that it can be described by a finite number of simple expressions called 
\emph{linear path schemes}.
Only in 2015, careful examination of these linear path schemes led to
exponential upper bound on the length of the shortest path, and 
consequently to \pspace upper bound~\cite{BlondinFGHM15}.
%
Concerning unary \dvass, polynomial upper bound on the length of the shortest path was shown
\cite{DBLP:conf/lics/EnglertLT16,DBLP:journals/jacm/BlondinEFGHLMT21},
thus yielding \nl-completeness. 

Our understanding of the model drops drastically for dimensions larger than 2,
as most of good properties admitted in dimension 1 or 2 vanish.
For instance, already since the seminal paper~\cite{DBLP:journals/tcs/HopcroftP79} it is known that
reachability sets of \tvass are not necessarily semi-linear.
Investigation of \tvass was advocated by many papers cited above, 
e.g.~\cite{DBLP:journals/tcs/HopcroftP79,BlondinFGHM15,DBLP:journals/jacm/BlondinEFGHLMT21},
but until now no specific complexity bounds for \tvass are known, except for generic 
parametric bounds known for \parvass d in arbitrary fixed dimension $d\geq 3$.
By~\cite{LS19}, the reachability problem in \parvass d
is in  $\F_{d+4}$,%
\footnote{The complexity class $\F_i$ corresponds to the $i$-th level 
of Grzegorczyk's fast-growing hierarchy~\cite{DBLP:journals/toct/Schmitz16}.}
later improved to $\F_d$~\cite{DBLP:conf/icalp/FuYZ24}.
In case of \tvass, this yields membership in $\F_3 = \tower$.
On the other hand, no lower bound is known for binary \tvass except for the \pspace lower bound 
inherited from binary \dvass
(for unary \tvass, \np-hardness has been recently shown by~\cite{DBLP:conf/focs/0001CMOSW24}).
The complexity gaps remains thus huge, namely between \pspace and \tower.
As our main result, we narrow down this gap significantly.


\para{Contribution}

In this paper we investigate complexity of the reachability problem in \tvass.
Our main result is the first elementary upper complexity bound for the problem:

\begin{theorem}\label{thm:expspace}
The reachability problem in \tvass is in \twoexpspace, under binary encoding.
\end{theorem}

\noindent
In particular, this refutes the natural conjecture that for every $d \geq 3$, the reachability problem for 
\parvass d is $\F_d$-complete and provides the first algorithm, which solves the problem
for \vass with finite reachability sets faster than exhaustive search.


Our way to prove Theorem~\ref{thm:expspace} is by bounding triply-exponentially 
the length of the shortest path between the given source and target configurations.
This main technical result, formulated in Lemma~\ref{lem:main} below,
applies to \emph{sequential} \tvass, which are sequences of strongly connected components
$V_1, \ldots, V_k$ linked by single transitions $u_1, \ldots, u_{k-1}$ 
(see Figure \ref{fig:kcompvass}; the rigorous definition is given in Section \ref{sec:def}).

\newcommand{\labelscvass}[1]{{\tiny\begin{tabular}{c}\tiny strongly\\ \tiny connected\\ \tiny \tvass $V_{#1}$\end{tabular}}}
\begin{figure}
\centering
\scalebox{1.0}{
\begin{tikzpicture}[shorten >=1pt,node distance=3cm,on grid,>={Stealth[round]},
    every state/.style={draw=blue!50,thick,fill=blue!20},scale=0.25]

  \node[state]          (q_0)                      {\labelscvass 1};
  \node[state]          (q_1) [right=of q_0] {\labelscvass 2};
  \node[]                  (q_ldots) [right=of q_1] {\ \ \ \ \ \ ... \ \ \ \ \ \ };
  \node[state]          (q_k) [right=of q_ldots] {\labelscvass k};

  \path[->] 
            (q_0) edge              node [above]  {$u_1$} (q_1)
            (q_1) edge              node [above]  {$u_2$} (q_ldots)
            (q_ldots) edge              node [above]  {$u_{k-1}$} (q_k);
\end{tikzpicture}
}
\caption{A sequential \tvass.}
\label{fig:kcompvass}
\end{figure}

Given a \vass $V$ and source and target configurations $s, t$, by
$\size(V,s,t)$ we mean the sum of absolute values of all the numbers occurring in transitions of $V$, $s$ and $t$, 
plus the number thereof.

\begin{lemma}\label{lem:main}
If there is a path from $s$ to $t$ in a sequential \tvass $V$,
then there is one of length at most 
$\kbound {\size(V,s,t)} {\OO(k)}$, where $k$ is the number of components of $V$.
\end{lemma}


\noindent
Therefore the length of the shortest path in a $k$-component \tvass is bounded by $\kbound M {\OO(k)}$,
where $M = \size(V,s,t)$ is the size of input, under unary encoding.
This is the first bound on the shortest path in \vass of dimension higher than 2, that is not based on
the size of (finite) reachability sets.
Indeed, in a \tvass of size $M$, the size of finite reachability sets may be an arbitrary high tower of exponentials.

In consequence of Lemma \ref{lem:main}, 
Theorem~\ref{thm:expspace} follows immediately: 
the upper bound of Lemma \ref{lem:main} 
is triple-exponential in the size of input, irrespectively whether unary or  binary encoding is used, 
which implies the same bound on norms of configurations
along the shortest path.
This yields a nondeterministic double-exponential space algorithm that
first guesses a sequence of components leading from the source state to the target one,
and then searches for a witnessing path.
Note that the complexity bound under unary encoding is not better than under binary encoding.


Lemma \ref{lem:main} immediately yields further upper bounds for the reachability problem,
when the number of components is fixed:

\begin{corollary}\label{cor:fixed-components}
For every fixed $k \geq 1$, the reachability problem in $k$-component \tvass is:
\begin{itemize}
  \item in \nl, under unary encoding,
  \item in \pspace, under binary encoding.
\end{itemize}
\end{corollary}

\noindent
Indeed, for every fixed $k\geq 1$, the bound of Lemma~\ref{lem:main} is polynomial with respect to
unary input size.
Therefore the length of the shortest path, as well as the norm of configurations along this path, are 
polynomially bounded in case of unary encoding,
and exponentially bounded in case of binary encoding.
This bounds yield membership in \nl and \pspace, respectively.
Thus for every fixed $k\geq 1$, the complexity of $k$-component \tvass matches the complexity of \dvass.

\para{Organisation of the paper}
We start by introducing notation and basic facts in Section~\ref{sec:def}.
Overview of the proof of our main result, Lemma \ref{lem:main}, is presented in Section~\ref{sec:overview}.
In Section \ref{sec:1comp} we focus on 1-component \tvass, thus establishing the base of induction
for Lemma \ref{lem:main}.
Next, in Section \ref{sec:tools} we introduce the fundamental concept of \sandwich sets,
and formulate our core technical result: reachability sets of \dvass
are \sandwich.
The result is then applied in the inductive proof of Lemma \ref{lem:main} in Section \ref{sec:mainproof}.
We conclude in Section~\ref{sec:future}.



\section{Preliminaries}\label{sec:def}

Let $\Q, \Qnonneg, \Qpos$ denote the set of all, nonnegative, and positive rationals, respectively,
and likewise let $\Z, \N, \Npos$ denote the respective sets of integers.
For $a, b \in \Z$, $a \leq b$, let $\setfromto a b$ denote the set $\{x \in \Z \mid a \leq x \leq b\}$.
The $j$th coordinate of a vector $w\in\Q^d$ we write as $w_j$.
Thus $w = (w_1, \ldots, w_d)$.
For $q\in \Q$,
by $\vec q$ we denote the constant vector $(q, \ldots, q) \in \Q^d$.

\para{Vector addition systems with states}

Let $d\in\Npos$.
A $d$-dimensional \emph{vector addition system with states} (\parvass d in short) $V=(Q,T)$ consists of
a finite set $Q$ of states, and a finite set of transitions $T\subseteq Q\times \Z^d \times Q$.
A \emph{configuration} $c$ of $V$ consists of a state $q\in Q$ and a nonnegative vector $w\in\N^d$,
and is written as $c=q(w)$.
A transition $u=(q,v,q')$ induces \emph{steps} $q(w) \trans{u} q'(w')$ between configurations, where $w' = w+v$.
We refer to the vector $v\in\Z^d$ as the \emph{effect} of the transition $(q,v,q')$ or of an induced step.
A \emph{path} $\pi$ in $V$ is a sequence of steps with the proviso that the target configuration of every step matches
the source configuration of the next one:
\begin{align} \label{eq:path}
\pi \ = \ c_0 \trans{u_1} c_1 \trans{} \ldots \trans{u_n} c_n.
\end{align}
The \emph{effect} $\eff(\pi)\in\Z^d$
of a path is the sum of effect of all steps, and its \emph{length} is the number $n$ of steps.
We say that the path is \emph{from} $c_0$ \emph{to} $c_n$, call $c_0, c_n$ source and target 
configuration, respectively, of the path, and write $c_0\trans{\pi}c_n$.
We also write $c \tran c'$ if there is some path from $c$ to $c'$.
A path $q(v) \tran q(v')$ in $V$ with the same source and target state
we call a \emph{cycle}.
A cycle is \emph{simple} if the only equality of states along the cycle is the equality of source and target states.
When dimension $d$ is irrelevant, we write \vass instead of \parvass d.

By the \emph{geometric} dimension of a \parvass d we mean the dimension of the vector
space $\Lin V \subseteq \Q^d$ spanned by effects of all its simple cycles.
%
In the sequel we most often consider \dvass and \tvass, but also \emph{\geomvass}, i.e.,
\tvass of geometric dimension at most 2.

Two paths $\pi$ and $\pi'$ can be \emph{concatenated} (composed), written $\pi; \, \pi'$, if
the target configuration of $\pi$ equals the source one of $\pi'$.
As long as it does not lead to confusion, 
we adopt a convention that when concatenating paths $\pi;\,\pi'$, 
the latter path $\pi'$ is silently \emph{moved}
so that its source matches the target of $\pi$, 
under assumption that the source state of $\pi'$ is the same as the target state of $\pi$.
For instance, we write $\pi^m$ to denote the $m$-fold concatenation of a cycle $\pi$, even if $\eff(\pi)\neq\vec 0$.

We use the following notation for measuring size of representation of a \vass.
By \emph{norm} of a vector $v=(v_1, \ldots, v_d) \in \Q^d$, denoted $\norm(v)$, 
we mean the sum of absolute values of all numbers appearing in it:
$\norm(v) = \absv{v_1} + \ldots + \absv{v_d}$; and 
by norm of a set of vectors $P$ we mean the sum of norms of its elements:
$\norm(P) = \sum\setof{\norm(v)}{v\in P}$. 
By \emph{norm} of a configuration $q(w)$, or of a transition $(q, w, q')$,
we mean the norm of its vector $w$.
By \emph{size} of a \vass $V$, denoted $\size(V)$, we mean
the sum of norms of all its transitions, 
plus the number of transitions $\card T$.
The \emph{norm} of a \vass is the maximal norm of its transition.

We often implicitly extend a \vass $V$ with source configuration $s$, or with
a pair of source and target configurations $s, t$.
Slightly overloading terminology, a pair $(V, s)$ and a triple $(V, s, t)$ we call a \vass too.
For convenience we overload further and put $\size(V,s) = \size(V) + \norm(s)$ and 
$\size(V, s, t) = \size(V) + \norm(s) + \norm(t)$.
The \emph{reverse} of a \vass $V = (Q, T)$ is defined as $\rev V = (Q, T')$,
where $T'$ is obtained by reversing all transitions in $T$:
$T' = \setof{(q', -v, q)}{(q,v,q')\in T}$.
Overloading the notation again, we put $\rev{(V, s, t)} := (\rev V, t, s)$.

Given a \vass together with an initial configuration $(V, s)$, we write 
$\reach(V,s)$ to denote the set of configurations $t$ such that $V$ has a path from $s$ to $t$.
For every state $q\in Q$ we write $\reach_q(V,s)$ to denote the
set of vectors $w\in\N^d$ such that $q(w) \in \reach(V,s)$.
We write shortly $\reach(s)$ and $\reach_q(s)$ if the \vass $V$
is clear from the context.
If $t\in\reach(s)$ we say that $t$ is \emph{reachable} from $s$.
If $t+\Delta\in\reach(s)$ for some $\Delta\in\N^d$, we say that $t$ is \emph{coverable} from $s$.

We consider a variant of \vass, called \zvass, where the nonnegativeness constraint is dropped.
Syntactically, \zvass is the same as \vass, namely consists of a finite set of states and a finite set of transitions
$(Q, T)$. 
Semantically, configurations of a \zvass are $Q \times \Z^d$, while
all definitions (path, reachability set, etc.) are the same as in case of \vass.
Equivalently, we may also speak of $\Z$-configurations and $\Z$-paths of a \vass, i.e., 
configurations and paths where the nonnegativeness constraint is dropped.
Note that every $\Z$-path $q(w) \trans{\pi} q'(w')$ may be \emph{lifted} to become a path
$q(w+\Delta) \trans{\pi} q'(w'+\Delta)$, for some $\Delta \in \N^d$.

\para{Sequential \vass}

We define the state graph of a \vass $V = (Q, T)$: nodes are states $Q$, and there is an edge
$(q,q')$ if $T$ contains a transition $(q, v, q')$ for some $v\in\Z^d$.
A \vass is called \emph{strongly connected} if its state graph is so.
A \vass $V = (Q, T)$ is called \emph{sequential},
if it can be partitioned into a number of strongly connected 
\vass $V_1=(Q_1, T_1), \ldots, V_k=(Q_k, T_k)$ with pairwise disjoint state spaces, 
and $k-1$ transitions
$u_i=(q_i, v_i, q'_i)$, for $i\in\setto{k-1}$, where $q_i\in Q_i$ and $q'_i \in Q_{i+1}$
(recall Figure \ref{fig:kcompvass}).
Thus $Q=Q_1\cup\ldots\cup Q_k$ and 
$T=T_1 \cup \ldots \cup T_k \cup \{u_1, \ldots, u_{k-1}\}$.
We call $V$ a \emph{$k$-component} sequential \vass,
or $k$-component \vass in short,
and write down succinctly as 
$
V  =  \ktvass.
$
The \vass $V_1, \ldots, V_k$ are called \emph{components}, and transitions
$u_1, \ldots, u_{k-1}$  \emph{bridges}.
By definition, a $1$-component sequential \vass is just a strongly connected \vass.

\para{Integer solutions of linear systems}
We will intensively use the following immediate corollary of \cite{Pottier91}
(see also \cite[Prop.~4]{taming}):
\begin{lemma}
[\cite{taming}, Prop.~4]
\label{lem:taming}
Consider a system $A\cdot x = b$ of $m$ Diophantine linear equations with $n$ unknowns,
where absolute values of coefficients are bounded by $N$.
Every pointwise minimal nonnegative integer solution has norm at most
$\OO(nN)^m$.
\end{lemma}

\para{Diagonal property}

We prove that
if there is a path in a \vass achieving a large value on every coordinate,
then there is a path of bounded length that achieves \emph{simultaneously} large values on all coordinates.%
\footnote{Lemma \ref{lem:sim-high-d} is inspired by \cite[Lemma 4.13]{LS19}, 
but the statement and the proof are different.}
We consider a general case of arbitrary dimension,
as we believe it is of an independent interest, but in the sequel we will use it only for dimension $d = 3$.

\begin{lemma}\label{lem:sim-high-d}
For every $d \in \N$ there are nondecreasing polynomials $P_d, R_d$ such that 
for every \parvass d $(V, s)$ of norm $N$, with $n$ states, and for every $U\in\N$,
if $V$ has a path from $s$ that for every $i\in\setto d$ contains a configuration 
$q(w_1, \ldots, w_d)$ with $w_i \geq P_d(n, N, U)$,
then $V$ has also a path $s \tran q(w_1, \ldots, w_d)$ of length at most $R_d(n, N, U)$
such that $w_i \geq U$ for every $i\in\setto d$.
\end{lemma}

\begin{appendixproof}[Proof of Lemma \ref{lem:sim-high-d}]
Below, we will refer to two inductively defined sequences $(H_i)_{i\in\Npos}$, $(L_i)_{i\in\Npos}$,
(implicitly) parametrised
by numbers $n, N, U \in \N$:
let $H_1 := U$, $L_1 = n U$, and for  $i > 1$ let $H_i = U + N L_{i-1}$ and  
$L_i = n (H_i)^i + L_{i-1}$.

\begin{lemma}\label{lem:sim-high}
Let $d \in \N$, and let $(V, s)$ be a \parvass d of norm $N$, with $n$ states.
If $V$ has a path from $s$ that for every $i\in\setto d$ contains a configuration 
$q(w_1, \ldots, w_d)$ with $w_i \geq H_d$,
then $V$ has also a path $s \tran q(w_1, \ldots, w_d)$ of length at most $L_d$
such that $w_i \geq U$ for every $i\in\setto d$.
\end{lemma}

\begin{proof}
The proof proceeds by induction over the dimension $d$, and 
follows the idea of Rackoff~\cite[Lemma 3.4]{DBLP:journals/tcs/Rackoff78}.
Let $V$ be a fixed \parvass d.

When $d = 1$, as $H_1 = U$ the first claim is immediate.
The bound on length $L_1 = n U$ is also obtained immediately, by removing repetitions of
configurations.

For the induction step, we assume that Lemma~\ref{lem:sim-high} holds for dimension $d-1$,
and show it for dimension $d$.
Consider the  shortest 
path $s\trans{\rho} u$ from $s$ in $V$ such that for every $i\in\setto d$, the path contains a configuration 
$q(w_{1}, \ldots, w_{d})$ with $w_{i} \geq H_d$.
Let $q(w_1, \ldots, w_d)$ be the first  configuration  on $\rho$ with
$(w_1, \ldots, w_d) \notin {\setfromto 0 {H_d-1}}^d$:
\[
s \trans {\rho_1} q(w_1, \ldots, w_d) \trans{\rho_2} u.
\]
\Wlog we may assume that $w_d \geq H_d$. 
The length of $\rho_1$ is at most $n (H_d)^d$, as the 
configurations along $\rho_1$ are bounded by $H_d -1$ and do not repeat.

Let $\essdvass V$ denote the \parvass {(d-1)} obtained by dropping the $d$th coordinate of $V$.
By the induction assumption, $\essdvass V$ has a path 
\begin{align*}
q(w_1, \ldots, w_{d-1}) \trans{\rho_3} p(v_1, \ldots, v_{d-1})
\end{align*}
of length at most $L_{d-1}$, whose target vector satisfies $v_i \geq U$ for all $i\in\setto {d-1}$.
Steps of $\essdvass V$ are steps of $V$ where the $d$th coordinate is dropped,
and each such step may decrease the value on $d$th coordinate by at most $N$.
Therefore $w_d\geq H_d=U+N L_{d-1}$ is large enough so that $L_{d-1}$ steps of 
$\rho_3$ yield at least $U$ on the $d$th coordinate.
Therefore $\rho_3$  can be traced back to a path 
\[
q(w_1, \ldots, w_{d}) \trans{\rho_4} p(v_1, \ldots, v_{d})
\]
in $V$, of the same length as $\rho_3$, where $v_d\geq U$.
The concatenated path $\rho_1; \, \rho_4$ has length at most
$L_d = n (H_d)^d + L_{d-1}$ and hence satisfies the claim of Lemma~\ref{lem:sim-high}.
\end{proof}

We show that $H_i$ and $L_i$ are bounded by a polynomial in $n, N, U$, of
degree doubly exponential in dimension $d$.
We concentrate on $H_i$, as
$H_i \leq L_i \leq H_{i+1}$.

\begin{proposition}\label{prop:hd-bound}
$H_i \leq (4Nn)^{2i!} U^{i!}$.
\end{proposition}

\begin{proof}
Let $C = 4Nn$.
As $H_i \leq 2 N L_{i-1}$ and $L_i \leq 2n (H_i)^i$, we have:
\begin{equation}\label{eq:hd-bound}
H_i \leq C \cdot (H_{i-1})^{i-1}.
\end{equation}
Instead of showing  $H_i \leq C^{2i!} U^{i!}$, we
prove a slightly stronger inequality:
\[
H_i \leq C^{2(i-1)! - 1} U^{(i-1)!},
\]
by induction on $i$.
When $i = 1$ we have $H_1 = U \leq C U$, and when
$i = 2$, by \eqref{eq:hd-bound}
we have $H_2 \leq C H_1 = C U$, as required.
The induction step follows by:
\[
H_{i+1} \leq C \cdot (H_{i})^{i} \leq C \cdot (C^{2(i-1)! -1} U^{(i-1)!})^{i} \leq C^{2 i! - 1} U^{i!},
\]
where first inequality is \eqref{eq:hd-bound}, 
the second one is the inductive assumption, and the latter one follows
since $1 + (2(i-1)! -1)\cdot i \leq 2i! -1$ when $i\geq 2$.
This completes the proof.
\end{proof}
Lemma~\ref{lem:sim-high}, Proposition~\ref{prop:hd-bound} and the inequality $L_i \leq H_{i+1}$
entail Lemma \ref{lem:sim-high-d}.
\end{appendixproof}

\para{Length-bound on shortest path}

A function $f:\N\to\N$ is called \emph{nondecreasing} if $f(n) \geq n$ for every $n \in \N$,
and $f(n)< f(m)$ for all $n<m$.
Functions used in the sequel are most often nondecreasing.
We say that a class $\C$ of \vass or \zvass is \emph{\lb} by a 
non-decreasing function $f:\N\to\N$ if 
for every $(V, s, t)$ in $\C$, if $s \tran t$ then 
$s \trans{\pi} t$ for some path $\pi$ of length at most $f(\size(V, s, t))$.
A class $\C$ which is \lb  by some nondecreasing
polynomial we call \emph{\plb}.
It is known that \dvass  have this property:
\begin{lemma}[\cite{DBLP:journals/jacm/BlondinEFGHLMT21}, Theorem 3.2]\label{lem:2vass-plb}
\dvass are \plb.%
\footnote{\cite{DBLP:journals/jacm/BlondinEFGHLMT21} adopts a slightly different, but equivalent
up to a constant multiplicative factor, definition of norm and 
size.
}
\end{lemma}
As a corollary of Lemmas \ref{lem:2vass-plb} and \ref{lem:taming}, respectively,
we derive the property for \geomvass 
(using \cite[Lemma 5.1]{Zhang-geom}) and \tzvass, respectively:
\begin{lemma}
\label{lem:geom-plb}
\Geomvass are \plb.
\end{lemma}
\begin{appendixproof}[Proof of Lemma \ref{lem:geom-plb}]
We rely on the construction of \cite[Lemma 5.1]{Zhang-geom}, which transforms a given 
\geomvass $(V, s, t)$ into a \dvass $(\essdvass V, \essdvass s, \essdvass t)$ such that%
\footnote{
In \cite{Zhang-geom} only zero source and target vectors are considered, but the construction
routinely extends to arbitrary such vectors. 
}
\[
\Len {\essdvass V} {\essdvass s} {\essdvass t} \ = \ \{3 \cdot n \mid n \in \Len V s t\},
\]
and the size of the \dvass is only polynomially larger than the size of the original \geomvass. 
Therefore, as \dvass are \plb due to Lemma \ref{lem:2vass-plb}, so are also \geomvass.
\end{appendixproof}
\begin{lemma} \label{lem:zvass-plb}
\tzvass are \plb.
\end{lemma}
%
\begin{appendixproof}[Proof of Lemma \ref{lem:zvass-plb}]
Let $(V, s, t)$ be a \tzvass such that $s\trans{\sigma} t$.
Let $V=(Q,T)$, $s=q(w)$ and $t=q'(w')$.
Let $M = \size(V, s, t) = \size(V) + \norm(s) + \norm(t)$.
We express a path as a solution of a Diophantine system of linear equations, and rely
on Lemma \ref{lem:taming}.

Let $Q'\subseteq Q$ and $T'\subseteq T$ be the subsets of states and transitions that appear in $\sigma$.
Simple cycles that use only transitions from $T'$ we call $T'$-cycles.
The $\Z$-path $s\trans{\sigma} t$ decomposes into a $\Z$-path $\sigma_0$ that visits all states of $Q'$,
plus a number of $T'$-cycles.
Choose the shortest such $\sigma_0$.
The $\Z$-path $\sigma_0$ visits each state at most $\card Q$ times, as otherwise it could be shortened,
and therefore its effect has norm at most $M^2$.
The effect $\Delta$ of each $T'$-cycle has norm at most $M$, as it contains no repetition of a transition.
We choose, for each such vector $\Delta$, one of $T'$-cycles with effect $\Delta$.
Let $\cal C$ be the set of chosen $T'$-cycles. 
Its size is at most $(2M+1)^3 \leq \OO(M^3)$.
We define a system $\cal U$
of 3 linear equations (one for each dimension), whose unknowns $x_\delta$ correspond to $T'$-cycles $\delta$ from $\cal C$:
\[
\sum_{\delta\in \mathcal{C}} x_{\delta} \cdot e_\delta \ + \ e_{\sigma_0} \ = \ t - s,
\]
where $e_\delta\in\Z^3$ is the effect of $\delta$, 
and $e_{\sigma_0}\in\Z^3$ is the effect of $\sigma_0$.
The system has a nonnegative integer solution, namely the one obtained from decomposition of $\sigma$
into $\sigma_0$ and simple cycles.
As all coefficients of $\cal U$ are bounded by $M^2$,
by Lemma \ref{lem:taming} the system has a solution of norm $\OO(M^3\cdot M^2)^3=\OO(M^{15})$.
The solution $(x_\delta)_\delta$ 
yields a $\Z$-path $s\tran t$ of length $\OO(M^{16})$, consisting of $\sigma_0$ with
attached all cycles $\delta\in\mathcal{C}$
(this is possible, as $\sigma_0$ visits all states used by the cycles),
each $\delta$ iterated $x_\delta$ times.
This completes the proof.
\end{appendixproof}
Lemma \ref{lem:main} states that there exists a constant $C\in\N$ such that
for every $k\in\N$, the $k$-component \tvass are 
length-bounded by the function $M\mapsto \kbound {M} {C\cdot k}$.
Therefore for every fixed $k\in\N$, the $k$-component \tvass are \plb,
even if the degree of polynomial grows doubly exponentially in $k$.


\section{Overview} \label{sec:overview}
In this section we present an overview of the proof of our main result, namely of Lemma~\ref{lem:main}.
The proof proceeds by an induction on the number $k$ of components in a sequential \tvass.
The main idea is that either the situation is easy (a short path can be obtained by lifting up a $\Z$-path)
or the first component can be transformed into a finite union of essentially two-dimensional \vass (more precisely,
finite union of \geomvass), each of size bounded polynomially. This transformation is shown in Lemma~\ref{lem:len-eq}.
The induction base is shown in Section~\ref{sec:1comp}. We present the proof of the one component case in detail,
as it illustrates the main concepts of the proof of Lemma~\ref{lem:len-eq}, but in a much simpler setting.
When the first component is transformed into essentially a \dvass, we can use the fact that
reachability sets in \dvass are semi-linear and the size of the semi-linear representation
is at most exponential~\cite{BlondinFGHM15} (the result is true as well in \geomvass).
This fact can be exploited to reduce reachability for $k$-component \tvass to reachability for $(k-1)$-component \tvass
of exponentially larger size, the details of this reduction are explained in Section~\ref{sec:tools}
in the paragraph about the idea of the proof of Lemma~\ref{lem:main}.
However, if we use semi-linear sets, the exponential blowup in unavoidable and this approach gives us a \tower algorithm
resulting from a linear number of exponential blowups (thus not better than \cite{DBLP:conf/icalp/FuYZ24}).
In order to improve the complexity we introduce a novel notion of suitable over- and under-approximations of semi-linear sets.
One of our key technical contributions is Lemma~\ref{lem:2vass-sandwich} stating that reachability sets
of \dvass can be well approximated. 
Intuitively speaking, the precision of the approximation has to be good enough for correctness of the inductive proof;
the better the precision, the bigger the representation of approximants gets.
This approach allows us to reduce reachability for $k$-component \tvass to reachability in $(k-1)$-component \tvass of size which is not anymore exponential, but is polynomial in $B$,
where $B$ is the minimal length of a path in $(k-1)$-component \tvass.
This means that $n^m$ bound on the minimal path length for $(k-1)$-component \tvass implies roughly a $n^{m^2}$
bound on the minimal path length for $k$-component \tvass. The transformation
$n^m \mapsto n^{m^2}$ applied linear number of times results in triply-exponential upper bound for the minimal length of a path in \tvass.


\section{1-component \tvass are \plb} \label{sec:1comp}

This section is devoted to the proof of the 
induction base for the proof of Lemma \ref{lem:main}:

\begin{lemma} \label{lem:1comp}
1-component \tvass are \plb.
\end{lemma}

\noindent
We also develop a framework to be exploited in the induction step in 
Section \ref{sec:mainproof}.
We may safely restrict to \tvass of geometric dimension 3,
as otherwise Lemma~\ref{lem:geom-plb} immediately implies Lemma~\ref{lem:1comp}.

\para{Case distinction}
A \tvass $(V, s, t)$, where $s=p(w)$ and $t=p'(w')$,
is \emph{forward-diagonal} if
$p(w) \trans{*} p(w+\Delta)$ in $V$ for some $\Delta \in (\Npos)^3$.
Symmetrically, $(V, s, t)$ is \emph{backward-diagonal} if $\rev{(V, t, s)}$ is diagonal, i.e., if 
$p'(w'+\Delta') \trans{*} p'(w')$ in $V$ for some $\Delta' \in (\Npos)^3$.
Finally, $V$ is \emph{diagonal} if it is both forward- and backward-diagonal.
Obviously, the vectors $\Delta$ and $\Delta'$ need not be equal in general.


Let  $E=\{e_1, \ldots, e_n\} \subseteq \Z^3$ be the effects of simple cycles of $V$.
We define the (rational) open cone generated by this set
to contain all positive rational combinations of vectors from $E$:
\[
\cone {V} \ = \ \setof{ r_1 \cdot e_1 + \ldots + r_n \cdot e_n }{r_1, \ldots, r_n \in \Qpos} \ \subseteq \ \Lin V.
\]
$\cone V$ is thus an open cone inside $\Lin V$.
A 1-component \tvass $V$ is called \emph{wide} if $(\Qpos)^3 \subseteq \cone V$, i.e.,
if $\cone V$ includes the whole positive orthant.

%

Let $\Len V s t$ denote the set of lengths of paths $s\tran t$ in $V$.
We need to
argue that there is a nondecreasing polynomial $Q$ 
such that every 1-component \tvass $(V, s, t)$
with a path $s\tran t$, has such path of length at most $Q(M)$, where 
$M=\size(V, s,t)$.
We split the proof into three cases:
\smallskip
\begin{enumerate}
\item
If $(V,s,t)$ is diagonal and wide, we exploit the fact that \tzvass are \plb, and use diagonality and wideness to
lift a short $\Z$-path into a path.
\item
If $(V,s,t)$ is diagonal but non-wide, we show that $(V,s,t)$ is \emph{length-equivalent} to a 
\geomvass $(\essdvass V, \essdvass s, \essdvass t)$ of polynomially larger size, namely 
$\Len V s t = \Len {\essdvass V} {\essdvass s} {\essdvass t}$.
\item
Finally, if $(V, s, t)$ is non-diagonal, we show that $(V,s,t)$ is \emph{length-equivalent} to a set
of three
\geomvass $\{(V_1, s_1, t_1), (V_2, s_2, t_2), (V_3, s_3, t_3)\}$, namely 
$\Len V s t = \Len {V_1} {s_1} {t_1} \cup \Len {V_2} {s_2} {t_2} \cup \Len {V_3} {s_3} {t_3}$
of polynomially larger size.
\end{enumerate}
\smallskip
\noindent
In the two latter cases we rely on the fact that \geomvass are \plb (Lemma \ref{lem:geom-plb}).
%
%
%
%
%
In consequence, $Q$ is to be the sum of polynomials claimed in the respective cases.
%
In the sequel let $(V, s, t)$ be a fixed 1-component \tvass  with $s\tran t$, where
$s=p(w)$ and $t=p'(w')$.

\para{Case 1. $(V, s, t)$ is diagonal and wide}
By diagonality, $p(w) \trans{\pi} p(w+\Delta)$ and
$p'(w'+\Delta')\trans{\pi'} p'(w')$  for some $\Delta,\Delta'\in(\Npos)^3$.

Let $P$ be a nondecreasing polynomial witnessing Lemma \ref{lem:zvass-plb}, i.e.,
\tzvass are \lb by $P$.
As there is a path $s \tran t$ in $V$, there is also a $\Z$-path $s\tran t$, and
by Lemma \ref{lem:zvass-plb} there is a $\Z$-path $s \trans{\sigma} t$ of length at most $P(M)$.
The maximal norm $N$ of $\Z$-configurations along $\sigma$ is thus bounded by $M \cdot P(M)$,
as every step may update counters by at most $M$.

By diagonality, the configuration $p(w+\vec 1)$ is coverable in $V$ from $s$,
and symmetrically the configuration $p'(w' + \vec 1)$ is coverable in $\rev V$ from $t$.
Due to the upper bound of 
Rackoff~\cite[Lemma~3.4]{DBLP:journals/tcs/Rackoff78},
there is a nondecreasing polynomial $R$ such that in every \tvass of size $m$, the length of a covering
path is at most $R(m)$.
Therefore the lengths of 
both paths $p(w) \trans {\pi} p(w+\Delta)$  and $p'(w'+\Delta')\trans {\pi'} p'(w')$ in $V$, 
where $\Delta, \Delta' \in (\Npos)^3$,
may be assumed to be at most $R(M)$.
We argue that there is a cycle 
from the source configuration $p(w)$ 
that increases $w$ by some multiplicity of $\Delta'$:
\begin{lemma} \label{lem:FP}
There is a path $p(w) \trans \delta p(w+ \ell \cdot\Delta')$ of length $R(M)^{\OO(1)}$,
for some $\ell\in\Npos$.
\end{lemma}
Before proving the lemma we use it to complete Case 1.
We build a path $p(w) \tran p'(w')$ by concatenating $3$ paths given below.
The first one is  $\delta$
given by Lemma \ref{lem:FP}.
Note that $\ell$ is necessarily also bounded by $R(M)^{\OO(1)}$.
We replace $\ell$ by its sufficiently large multiplicity to enforce $\ell \geq M\cdot P(M)$, which
makes the length of $\delta$ and $\ell$ only bounded by $P(M) \cdot R(M)^{\OO(1)}$.
The multiplicity guarantees that the $\Z$-path 
$
p(w) \trans{\sigma} p'(w'), 
$
lifted by $\ell \cdot \Delta'$, becomes a path:
\[
p(w+\ell \cdot \Delta') \trans{\sigma} p'(w' + \ell\cdot\Delta'), 
\]
The length of $\sigma$ is bounded by $P(M)$.
Finally, let $\delta' = (\pi')^\ell$ be the $\ell$-fold concatenation of
the cycle $\pi'$:
\[
p'(w'+\ell\cdot\Delta') \trans{\delta'} p'(w').
\]
The length of this path is bounded by $\ell\cdot R(M) \leq P(M)\cdot R(M)^{\OO(1)}$.
We concatenate the three paths, $\tau := \delta; \, \sigma;\, \delta'$, to get a required path
\[
p(w) \trans{\tau} p'(w')
\]
of length bounded by $P(M) \cdot R(M)^{\OO(1)}$.
It thus remains to prove Lemma \ref{lem:FP} in order to complete Case 1.

\begin{proof}[Proof of Lemma \ref{lem:FP}]
Let $\pi_1$ be a cycle that visits all states (it exists since the considered \vass is strongly connected), and
let $\Delta_1\in\Z^3$ be its effect.
%
%
%
Relying on $\Delta\in(\Npos)^3$,
take a sufficiently large multiplicity $m\in\Npos$ so that 
$\widetilde \pi = \pi^m;\, \pi_1$ is a path with nonnegative effect.
The path $\widetilde \pi$ is a cycle and  its effect is
$\widetilde \Delta = m\cdot\Delta + \Delta_1 \in \N^3$.

As $\Delta'\in(\Npos)^3$, there is $\ell'\in\Npos$ such that 
$\ell'\cdot\Delta'-\widetilde \Delta \in (\Qpos)^3$,
and hence, by wideness of $(V, s)$, we have
$\ell'\cdot\Delta'-\widetilde \Delta \in \cone V$, namely
\[
\ell'\cdot\Delta' - \widetilde \Delta = r_1 \cdot e_1 + \ldots + r_n \cdot e_n
\] 
for some positive rationals $r_1, \ldots, r_n\in \Qpos$.
By Carathéodory's Theorem \cite[p.94]{cara}, $\ell'\cdot\Delta'-\widetilde \Delta$ is a combination of some $3$ vectors
among $e_1, \ldots, e_n$, say $e_1, e_2, e_3$:
\[
\ell'\cdot\Delta' - \widetilde \Delta \ = \ r_1 \cdot e_1 + r_2 \cdot e_2 + r_3 \cdot e_3,
\] 
for some positive rationals $r_1, r_2, r_3\in\Qpos$.
Therefore, the system of $3$ equations
\[
\ell \cdot (\ell'\cdot\Delta'-\widetilde\Delta)  \ = \ r_1 \cdot e_1 + r_2 \cdot e_2 + r_3 \cdot e_3,
\]
with unknowns $\ell, r_1, r_2, r_3$, has a positive integer solution.
We rewrite the system to:
\begin{align} \label{eq:syst_1comp}
\ell \ell'\cdot \Delta' - \ell m \cdot \Delta \ = \  \ell\cdot\Delta_1 + r_1 \cdot e_1 + r_2 \cdot e_2 + r_3 \cdot e_3.
\end{align}
%
Let $\sigma_i$ be simple cycle of effect $e_i$, for $i\in\setto 3$.
Let $
\sigma
$
be a $\Z$-path that starts (and ends) in state $p$ and consists of 
$\ell$-fold concatenation of the cycle $\pi_1$, with attached 
$(r_1)$-fold concatenation of $\sigma_1$,
$(r_2)$-fold concatenation of $\sigma_2$,
and
$(r_3)$-fold concatenation of $\sigma_3$
(since $\pi_1$ visits all states, this is possible).
The effect of $\sigma$ is the right-hand side of \eqref{eq:syst_1comp}, and therefore
$\sigma$ is a $\Z$-path from $p(w+\ell m \cdot \Delta)$ to
$p(w + \ell \ell'\cdot \Delta')$:
\[
p(w + \ell m \cdot \Delta) \trans{\sigma} p(w+\ell \ell'\cdot \Delta').
\]
It need not be a path in general, and therefore we are going to lift it.
Let $k\in\Npos$ be a multiplicity large enough so that $\sigma$ becomes a path
when lifted by $(k-1)\ell m \cdot \Delta$, i.e., when starting in $p(w+k \ell m \cdot \Delta)$, 
and also becomes a path when lifted by $(k-1)\ell\ell'\cdot \Delta'$, i.e., when
ending in 
$p(w+k \ell\ell' \cdot \Delta')$.
In this case, the $k$-fold concatenation of $\sigma$ is also a path:
\begin{align} \label{eq:thepath}
p(w+k\ell m\cdot \Delta) \trans{\sigma^{k}} 
p(w+k \ell\ell'\cdot \Delta'),
\end{align}
since all points visited in the inner iterations of $\sigma$ are bounded 
from both sides by corresponding points visited in the first and the last iteration of $\sigma$.
Precomposing this path with $p(w)\trans{\pi^{k\ell m}} p(w + k\ell m\cdot \Delta)$ yields a path
$p(w) \tran p(w+ k \ell \ell'\cdot \Delta')$, as required.

\smallskip

We now (roughly) bound the magnitudes of all items involved in the above reasoning
by a constant power of $R(M)$.
\Wlog we may assume that the cycle $\pi_1$ uses every transition at most $\card Q\leq M$ times,
and thus both the length and norm of the effect of $\pi_1$ are bounded by $M^2$.
In consequence, $\pi_1$ may decrease counters by at most $M^2$.
Therefore $m\leq M^2$, and
norms of vectors $\Delta, \Delta', \widetilde\Delta$ are all bounded by 
$\OO(M^3\cdot R(M))$. 
The effects $e_1, e_2, e_3$ of simple cycles $\sigma_1, \sigma_2, \sigma_3$ are at most $M$, 
as no transition repeats along a simple cycle.
Therefore by Lemma \ref{lem:taming}, the system \eqref{eq:syst_1comp} has a solution $(\ell, r_1, r_2, r_3)$ 
of norm at most $D=\OO(M^3\cdot R(M))^3 = \OO(M^9 \cdot R(M)^{3})$, 
and $\sigma$ has length at most $M^2 \cdot D$
(since $\pi_1$ has length at most $M^2$). 
Therefore we deduce $k\leq M^3 \cdot D$, 
and hence
the path \eqref{eq:thepath} has length at most $M^5 \cdot D^2$. 
In consequence of the above bounds, $k \ell m \leq M^5 \cdot D^2$, 
and the final path $p(w) \tran p(w+ k \ell v)$ has length at most 
$R(M) \cdot M^5 \cdot D^2
\leq \OO(R(M)^{30})\leq R(M)^{\OO(1)}$.
\end{proof}

\para{Case 2. $(V, s, t)$ is non-wide}

%
Every non-zero vector $a=(a_1, a_2, a_3)\in\Z^3$ defines an open half-space
\[
H_a \ = \ \setof{x\in\Q^3}{\innprod a x > 0},
\]
where $\innprod a x = a_1 x_1 + a_2 x_2 + a_3 x_3$ stands for the inner product of 
$x =(x_1, x_2, x_3)$ and $a$.
As $V$ is assumed to be of geometric dimension 3,
$\cone V$  is an intersection of open half-spaces:
\begin{claim} \label{claim:a}
$\cone V$ is an intersection of 
finitely many open half-spaces $H_a$, with $\norm(a) \leq D := \OO(M^2)$.
\end{claim}
%

\begin{proof}
Norms of vectors generating $\cone V$ --- i.e., effects of simple cycles --- are at most $M$, 
as no transition repeats along a simple cycle.
Consider vectors $a$ orthogonal to some of the facets of $\cone V$, i.e.,
orthogonal to two of the vectors generating $\cone V$. 
The vector $a$ is thus an integer solution of a system of 2 linear equations with 3 unknowns,
where absolute values of coefficients are bounded by $M$.
By Lemma \ref{lem:taming}, there is such an integer solution with $\norm(a)\leq \OO(M^2)$.
This completes the proof.
\end{proof}

As $V$ is non-wide, 
due to Claim~\ref{claim:a} we know that $\cone V$ is a \emph{non-empty} intersection of half-spaces $H_a$.
Therefore for some of these $H_a$ we have
$\cone V \subseteq H_a$, i.e.,
%
%
all points $x\in\cone V$ have positive inner product $\innprod a x > 0$.
%
This implies that
the value of  inner product with $a$ may not decrease along any cycle in $V$:
\begin{claim}\label{clm:inner_with_effect}
The effect $\delta\in\Z^3$ of every simple cycle has nonnegative inner product $\innprod a \delta \geq 0$.
\end{claim}
In consequence, on every path $s \tran t$ the value of inner product with $a$ is polynomially
bounded:
\begin{claim} \label{claim:ax}
Every configuration $q(x)$ on a path from $s$ to $t$ satisfies
$-B \leq \innprod a x \leq B$, where $B := \OO(M\cdot D)$. 
\end{claim}
\begin{proof}
Let $s=p(w)$ and $t=p'(w')$, and let $b = \innprod a w$ and $b' = \innprod a {w'}$.
Every path $s \tran t$ may be decomposed into simple cycles, whose effect may only preserve or increase
the inner product with $a$, plus a short path without cycles, and hence without repetitions of a transitions.
The effect of the latter path is thus in $[-M,M]$.
Therefore, as inner product may at most multiply norms, 
every configuration $q(x)$ on a path from $s$ to $t$ satisfies
$b-M \cdot D \leq \innprod a x \leq b'+M \cdot D$.
Knowing that $\norm(a)\leq D$, $\norm(w), \norm(w')\leq M$, 
and that inner product may at most multiply norms, 
by Claim \ref{claim:a} we deduce 
\[
-M\cdot D \leq b, b' \leq M\cdot D,
\]
and therefore
$-2\cdot M\cdot D \leq \innprod a x \leq 2\cdot M \cdot D$,
which implies the claim.
\end{proof}

We define a \geomvass $\essdvass V = (\essdvass Q, \essdvass T)$ 
by extending states with the possible values of inner product with $a$
(bounded polynomially by Claim \ref{claim:ax}).
We call $\essdvass V$ the \emph{\mytrim {$a$} {$B$}} of $V$. 
The set of states $\essdvass Q$ contains states of the form $\pair q b$, where $q\in Q$ and $-B \leq b \leq B$,
with the intention that every configuration
$c = q(x)$ of $V$ has a corresponding configuration 
$\essdvass c = \pair q b (x)$ in $\essdvass V$, 
where $\innprod a x = b$.
%
Therefore, for each transition $(q, v, q') \in T$
and for all $b, b'\in \setfromto{-B} B$
such that $b + \innprod a v = b'$,
we add to $\essdvass T$ the transition
\begin{align} \label{eq:tranV}
\big(\pair q b, v, \pair {q'} {b'}\big).
\end{align}

\begin{claim}\label{clm:trim_makes_geom}
$\essdvass V$ is a \geomvass.
\end{claim}

\begin{proof}
By construction, the effect of each cycle in $\essdvass V$ is orthogonal to $a$, and therefore 
$\Lin {\essdvass V}$ is included in a $2$-dimensional vector space.
\end{proof}
Relying on Claim \ref{claim:ax}, paths $s\tran t$ in $V$ have corresponding paths in $\essdvass V$,
and hence we get:
\begin{claim} \label{claim:lenlen}
$\Len V s t = \Len {\essdvass V} {\essdvass s} {\essdvass t}$.
\end{claim}
%
%
Finally, we argue that the size of $\essdvass V$ is bounded polynomially with respect to the size of $V$:
\begin{claim} \label{claim:sizeVV}
$\size(\essdvass V) \leq R(M) = \OO(M\cdot B).$
\end{claim}
\begin{proof}
Recalling Claims \ref{claim:a} and \ref{claim:ax}, namely
$\norm(a)\leq D$, $B=\OO(M\cdot D)$, we deduce that 
transitions \eqref{eq:tranV} contribute at most $(2B+1)M \leq \OO(M\cdot B)$ to the size of $\essdvass V$.
\end{proof}
We are now prepared to complete Case 2.
Let $P$ be the polynomial witnessing Lemma \ref{lem:geom-plb}, i.e.,
\geomvass are \lb by $P$.
As $V$ has a path $s\tran t$, 
By Claim \ref{claim:lenlen}, $\essdvass V$ has a path $\essdvass s \tran \essdvass t$.
By Lemma \ref{lem:geom-plb}, $\essdvass V$ has thus a path $\essdvass s \tran \essdvass t$
of length at most $P(\size(\essdvass V))$, i.e., relying on Claim \ref{claim:sizeVV},
of length at most $P(\OO(M^2\cdot D)) = \OO(P(M^{4}))$.
By Claim \ref{claim:lenlen} again, we get a path $s\tran t$ in $V$ of length
$\OO(P(M^{4}))$.
This completes Case 2.

\para{Case 3. $(V, s, t)$ is non-diagonal}

\Wlog assume that $(V, s)$ is not forward-diagonal (otherwise replace $V$ by $\rev V$).
Therefore for all states $q$ the configuration $s'=q(w+(\vec{M{+}1}))$ is not
coverable from $s=p(w)$. 
Indeed, using strong-connectedness of $V$, if $s'$ were coverable from $s$ then
the configuration $p(w + \vec 1)$ would be coverable form $p(w)$, by 
extending the covering path of $s'$ with an arbitrary shortest path back to state $p$ 
(that cannot decrease a counter by more than $M$), which would contradict forward non-diagonality.

By Lemma \ref{lem:sim-high-d},
in every path from $s$, some coordinate
$j\in\setto 3$ is bounded by $B:=P_3(M, M, M+1)$ (we take $M$ as an upper bound for $n$ and $N$,
relying on $P_3$ being nondecreasing, and take $U=M+1$).
This property allows us, intuitively speaking, to describe all the paths of $V$ by paths of three
\geomvass $V_j$, for $j\in\setto 3$, where $V_j$ behaves exactly like $V$ except that dimension $j$ is additionally kept in state.
%
Formally, let $V_j := (Q_j, T_j)$, where 
\begin{align*}
Q_j = & \setof{\pair q b}{q\in Q, \ b \in \setfromto 0 B} \\
T_j =  & \setof{(\pair q b, v, \pair {q'}{b'})}{(q, v, q')\in T, \ b' = b + v_j}.
\end{align*}
The source and target configurations in $V_j$ are $s_j = \pair p {w_j}(w)$ and $t_j = \pair {p'}{w'_j}(w')$,
and there is a tight correspondence between paths in $V$ and paths in $V_1, V_2, V_3$:
\begin{claim} \label{claim:nondiag2vass}
$\Len V s t = \Len {V_1} {s_1} {t_1} \cup \Len {V_2} {s_2} {t_2} \cup \Len {V_3} {s_3} {t_3}$.
\end{claim}
The size of each of $V_j$ is bounded polynomially with respect to the size of $V$:
\begin{claim} \label{claim:nondiagsize}
The size of each of $V_j$ is at most $R(M)=\OO(M\cdot B)$.
\end{claim}
Let $P$ be the polynomial witnessing Lemma \ref{lem:geom-plb}.
As $p(w) \tran p'(w')$ in $V$,
by Claim \ref{claim:nondiag2vass}
there is a path $\pair p {w_j}(w) \tran \pair{p'}{w'_j}({w'})$ in $V_j$
for some $j\in\setto 3$.
Therefore, by Lemma \ref{lem:geom-plb} there is such a path of length 
at most $P(R(M))$ which,
again using Claim \ref{claim:nondiag2vass}, implies a path $p(w) \tran p'(w')$ in $V$ of the same length.
This polynomial bound completes Case 3, and hence also the proof of
Lemma \ref{lem:1comp}.
%


\section{\Sandwich reachability sets} \label{sec:tools}
In this section we introduce the crucial concept of \emph{\sandwich} sets.
In order to motivate it, we start by sketching the overall idea of the
proof of Lemma \ref{lem:main} (given in Section \ref{sec:mainproof}). 

Given a finite set $P\subseteq \N^d$ and $B\in\N$, we set:
\begin{align*}
P^* \ =  \ & \setof{p_1 + \ldots + p_k}{k\geq 0, \ p_1, \ldots, p_k \in P} \\
P^{\leq B} \ =  \ & \setof{p_1 + \ldots + p_k}{B\geq k\geq 0, \ p_1, \ldots, p_k \in P}.
\end{align*}
Sets of the form $b + P^* = \setof{b+p}{p\in P^*}$, for $b\in\N^d$ and finite $P\subseteq\N^d$,
are called \emph{linear},
and finite unions of linear sets are called \emph{semi-linear}.

\para{Idea of the proof of Lemma \ref{lem:main}}

Let $V=\ktvass$ be a $k$-component \tvass  that has a path $s=q(w) \tran q'(w')=t$.
If $V$ is diagonal and wide,
we use the pumping cycles $q(w) \tran q(w+\Delta)$ in $V_1$ and
$q'(w'+\Delta') \tran q'(w')$ in $V_k$ to lift a $\Z$-path $s\tran t$, \plb
due to Lemma \ref{lem:zvass-plb}, until it becomes a path
(as in Case 1 in the proof of Lemma \ref{lem:1comp}).

On the other hand, if $V$ is non-diagonal or non-wide, our strategy is to reduce the number of components
by 1, and to rely on the induction assumption for $k-1$,
by replacing the first component $V_1$ 
by one of finitely many \geomvass 
(as in Cases 2 and 3 of the proof of Lemma \ref{lem:1comp}).
Relying on the fact that the reachability sets in a \geomvass are semi-linear \cite{DBLP:conf/icalp/FuYZ24},
the proof could go as follows (yielding however only the
already known \tower upper bound \cite{DBLP:conf/icalp/FuYZ24}).
Using any of the linear sets $L = a +P^*$
describing the set $\reach_{q_2}(V,s)$, where $q_2$
is the source state of the second component $V_2$,
transform $V$ into a $(k-1)$-component \tvass $V'$ by dropping the first component $V_1$
and the first bridge $u_1$, and by adding to the remaining $(k-1)$-component
\tvass $(V_2)u_2 \ldots u_{k-1}(V_k)$ the self-looping transitions
$
(q_2, r, q_2),
$
one for every period $r\in P$.
The source configuration of $V'$ is $s' = q_2(a)$, i.e., its vector is the base of $L$.
The transformation preserves behaviour of $V$.
In one direction, a path $s\tran q_2(x) \tran t$ in $V$ crossing through $q_2(x)$ for 
some  $x\in L$ has a corresponding path $q_2(a) \tran q_2(x)\tran t$ in $V'$.
Conversely, each path $q_2(a) \tran q_2(x) \tran t$ in $V'$
gives rise to a path $s\tran t$ in $V$, by replacing executions of the self-looping transitions
$q_2(a) \tran q_2(x)$
(\mywlog executed in the beginning), by a path $s\tran q_2(x)$ in $V_1$,
bounded polynomially due to Lemma \ref{lem:1comp}.
However, $\size(V')$ may blow-up exponentially with respect to $\size(V)$, as bases and periods
of $L$ are only bounded exponentially, 
and therefore
this approach could only yield a $k$-fold exponential bound on the length of the shortest path
in $k$-component \tvass.


\Sandwich sets are designed as a remedy against the $k$-fold exponential blowup.
The idea is to measure the norms of base and periods of a semi-linear set $L$ 
parametrically with respect to, intuitively speaking, 
the prospective length $B$ of a path $s'\tran t$ in $V'$.
This allows us to control the blow-up of size of $V'$,
also parametrically with $B$, but requires
going outside of semi-linear sets and considering their \emph{$B$-approximations},
namely sets sandwiched between $a+P^{\leq B}$ and $a+P^*$,
good enough for correctness of the above-described transformation
of $V$ to $V'$.
As the outcome, the exponent of our bound on length of the shortest path in $k$-component \tvass
is, roughly speaking, square of the exponent of the respective bound in $(k-1)$-component \tvass. 
For $k$-component \tvass this yields exponent doubly-exponential in $k$, and hence the bound triply-exponential in $k$.
The rigorous reasoning is given in the proof of Lemma \ref{lem:ne} in Section \ref{sec:mainproof}.

\para{\Sandwich sets}

Let $A, B\in \N$.
By a $B$-\emph{approximation}
of a linear set $a + P^*\subseteq \N^d$ we mean any set $S\subseteq \N^d$ satisfying 
$
a + P^{\leq B} \subseteq S \subseteq a+P^*.
$
%
A set $X\subseteq \N^d$ is \emph{\kanapka {$A$} {$B$}}
if it is a finite union of:

\begin{itemize}
  \item linear sets $a + P^* \subseteq \N^d$ with $\norm(a) \leq B \cdot A$ and
  $\norm(P) \leq A$; and
  \item $B$-approximations of linear sets $a + P^* \subseteq \N^d$ with $\norm(a) \leq A$ and
  $\norm(P) \leq A$.
\end{itemize}
%

\noindent
Thus $X$ either includes 
$B$-approximation of a linear set $a+P^*$, whose norm of base is bounded by $A$,
or $X$ includes 
a whole linear set $a+P^*$, whose norm of base is only bounded by $B\cdot A$.
In both cases, norms of periods are bounded by $A$.

We say that a class $\C$ of \parvass d is \emph{\parsandwich{$F$}} for a function $F: \N \to \N$,
if for every \vass $(V, s)$ in $\C$, its state $q$,
and $B\in\N$,
the set $\reach_q(V,s)$ is \kanapka {$F(M)$} {$B$}, where $M=\size(V,s)$.
The class $\C$ is \emph{\sandwich} if it is \parsandwich{$F$}, for some nondecreasing 
polynomial $F$.
%
%


\begin{figure}[t]
\vspace{-1.5cm}
\begin{minipage}{0.34\textwidth}
\begin{tikzpicture}[scale=0.25]
\usetikzlibrary{automata, positioning}
\scalebox{0.65}{
\node[state] (q1) {$q_1$};
\node[state, right=of q1] (q2) {$q_2$};
\node[state, right=of q2] (q3) {$q_3$};
\node[state, right=of q3] (q4) {$q_4$};

\path[->] (q1) edge [loop above] node[above] {$(-1,2)$} (q1) edge node[above] {$(0,0)$} (q2); 
\path[->] (q2) edge [loop above] node[above] {$(2,-1)$} (q2) edge node[above] {$(0,0)$} (q3);
\path[->] (q3) edge [loop above] node[above] {$(-1,2)$} (q3) edge node[above] {$(0,0)$} (q4);
\path[->] (q4) edge [loop above] node[above] {$(2,-1)$} (q4);
}
\end{tikzpicture}
\end{minipage}
\begin{minipage}{0.32\textwidth}
\begin{tikzpicture}[scale=0.35]
\scalebox{0.65}{
\draw[->, thick] (0, 0) -- (18, 0) node[right] {$x_1$};
\draw[->, thick] (0, 0) -- (0, 10) node[above] {$x_2$};

\draw[step=1, gray, thin] (0, 0) grid (17, 9);

\foreach \x in {1,4,7,10,13,16} \fill[red] (\x,0) circle (7pt);
\foreach \x in {2,5,8,11,14} \fill[red] (\x,1) circle (7pt);
\foreach \x in {0,3,6,9,12} \fill[red] (\x,2) circle (7pt);
\foreach \x in {1,4,7,10} \fill[red] (\x,3) circle (7pt);
\foreach \x in {2,5,8} \fill[red] (\x,4) circle (7pt);
\foreach \x in {0,3,6} \fill[red] (\x,5) circle (7pt);
\foreach \x in {1,4} \fill[red] (\x,6) circle (7pt);
\foreach \x in {2} \fill[red] (\x,7) circle (7pt);
\foreach \x in {0} \fill[red] (\x,8) circle (7pt);

\draw[->] (1,0) -- (0,2) -- (2,1) -- (4,0) -- (3,2) -- (2,4) -- (1,6) -- (0,8) -- (2,7) -- (4,6) -- (6,5) -- (8,4) -- (10,3) -- (12,2) -- (14,1) -- (16,0);
}
\end{tikzpicture}
\end{minipage}
\begin{minipage}{0.32\textwidth}
\begin{tikzpicture}[scale=0.35]
\scalebox{0.65}{
\draw[->, thick] (0, 0) -- (18, 0) node[right] {$x_1$};
\draw[->, thick] (0, 0) -- (0, 10) node[above] {$x_2$};

\draw[step=1, gray, thin] (0, 0) grid (17, 9);

\foreach \x in {1,4,7,10,13,16} \fill[red] (\x,0) circle (7pt);
\foreach \x in {2,5,8,11,14} \fill[red] (\x,1) circle (7pt);
\foreach \x in {0,3,6,9,12} \fill[red] (\x,2) circle (7pt);
\foreach \x in {1,4,7,10} \fill[red] (\x,3) circle (7pt);
\foreach \x in {2,5,8} \fill[red] (\x,4) circle (7pt);
\foreach \x in {0,3,6} \fill[red] (\x,5) circle (7pt);
\foreach \x in {1,4} \fill[red] (\x,6) circle (7pt);
\foreach \x in {2} \fill[red] (\x,7) circle (7pt);
\foreach \x in {0} \fill[red] (\x,8) circle (7pt);

\draw[->,blue,thick] (1,0) -- (4,0);
\draw[->,blue,thick] (1,0) -- (1,3);

\draw[->,blue,thick] (2,1) -- (5,1);
\draw[->,blue,thick] (2,1) -- (2,4);

\draw[->,blue,thick] (0,2) -- (3,2);
\draw[->,blue,thick] (0,2) -- (0,5);
}
\end{tikzpicture}
\end{minipage}
\caption{Left: 4-component \dvass $V_2$. 
Middle: the set $\reach_{q_4}(V_2, q_1(1,0))$ and a path $q_1(1,0) \tran q_4(16,0)$.
Right: bases 
and periods 
 of an over-approximating semi-linear set $A+P^*$.}
\label{fig:zigzag}
\end{figure}
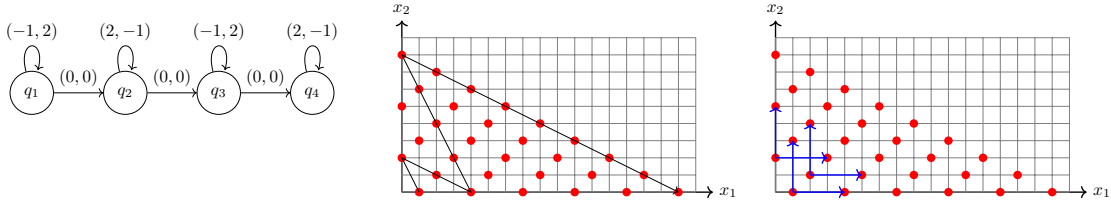

\begin{example}
For $k\geq 1$, let $V_k$ be a $(2k)$-component \dvass, where each component has just one state $q_i$
and one transition:
$(q_i, (-1,2), q_i)$ for odd $i$, and $(q_i, (2,-1), q_i)$ for even $i$.
Bridge transitions are $(q_i, (0,0), q_{i+1})$.
Figure~\ref{fig:zigzag} shows $V_2$ (left) and 
a path in $V_2$ from $s = q_1(1,0)$ to $t = q_4(16,0)$ together with 
the reachability set $\reach_{q_4}(V_2, s)$ (middle).
In general,
\begin{align} \label{eq:reachk}
X_k := \reach_{q_{2k}}(V_k, s) \ = \ \set{(x_1,x_2) \mid x_1+2x_2 \leq 4^k, \  x_1+2x_2 \equiv 1 \!\! \mod 3}.
\end{align}
Even if the size of the reachability set is 
exponential in $k$, for small $(x_1, x_2)$ it is periodic and the periods are small.
The set $X_k$ can be over-approximated by $A + P^*$ for $A = \set{(1,0),(2,1),(0,2)}$ and $P = \set{(0,3),(3,0)}$
(shown on the right of Figure~\ref{fig:zigzag}), namely for every $k\geq 1$ and $B\in\N$,
the set $X_k$ is \kanapka {$8$} {$B$}. 
For illustration, consider $Y := X_k \cap ((1,0) + P^*)$.
If $(1,0) + P^{\leq B} \subseteq X_k$ then $Y$ is a $B$-approximation
of $(1,0) + P^*$ with $\norm((1,0)), \norm(P) \leq 3 \leq 8$. 
Otherwise, there is some $(v_1, v_2) \in \big((1,0) + P^{\leq B}\big)\setminus X_k$, and
then $B$ is larger than $4^k$:
\[
4^k < v_1 + 2 v_2 \leq 2(v_1 + v_2) \leq 2(1+3B) \leq 8B.
\]
Therefore by \eqref{eq:reachk}, each $(x_1,x_2) \in Y$ satisfies 
$\norm(x_1,x_2) = x_1 + x_2 \leq x_1 + 2x_2 \leq 4^k < 8B$, and thus
$Y$, seen as a union of singletons, is a union of 
linear sets with norm of base bounded by $8B$ and empty set of periods. 
In both cases, 
$Y$ is \kanapka {$8$} {$B$}. 
\end{example}

In our subsequent reasoning we rely on the core technical fact:
%
\begin{lemma}\label{lem:2vass-sandwich}
\dvass are \sandwich.
\end{lemma}

We sketch here the intuition behind Lemma~\ref{lem:2vass-sandwich}, since
it is one of our main technical contributions. We first show that it is enough to show Lemma~\ref{lem:2vass-sandwich}
for a simple class of \dvass, called \dslps, which are of the form $\alpha_0 \beta_1^* \alpha_2 \ldots \alpha_{k-1} \beta_k^* \alpha_k$,
where $\alpha_i$ are fixed sequences of transitions and the loops $\beta_i$ are single transitions. This reduction uses standard techniques,
namely Theorem 1 in~\cite{BlondinFGHM15} stating that the reachability relation of a \dvass can be expressed as a union of reachability relations of a \dlps (\dlps are \dslps without the assumption that $\beta_i$ are single transitions) and Theorem 15 in~\cite{DBLP:conf/lics/EnglertLT16} providing the reduction from \dlps to \dslps. Next, we simplify the \dslps even more, using 
Theorem 4.16 in~\cite{DBLP:conf/focs/0001CMOSW24}, which states that any two vectors reachable by an \dslps can be reached
also by a path of the \dslps of a special form: except a short prefix and suffix it zigzags all the time between configurations close to vertical axis to configurations close to horizontal axis. Thus, to prove Lemma~\ref{lem:2vass-sandwich}
it essentially remains to show polynomial approximability for zigzagging paths.
To achieve that, we roughly speaking
investigate how application of two consecutive loops $\beta \in \N_+ \times \N_-$ and $\beta' \in \N_- \times \N_+$
affects the set of reachable configurations on some vertical line close to the vertical axis. We show that arithmetic sequence
is transformed into a finite union of arithmetic sequences such that the difference is kept at most polynomial in size of the \dslps
and the first term grows additively by at most a polynomial value. All that allows us to conclude that the set of vectors
reachable by zigzagging paths is a union of sets of a form similar to $a + Q^* + P^{\leq T}$ for some $a$, $Q$, $P$ and $T$.
This quite easily implies polynomial approximability.



\begin{appendixproof}[Proof of Lemma~\ref{lem:2vass-sandwich}]
We extend the definition of nondecreasing functions to many-argument ones:
a function $f: \N^k \to \N$ is \emph{nondecreasing} if it is monotonic
in every argument and $f(n_1, \ldots, n_k) \geq n_1 + \ldots + n_k$.
In the sequel we often bound certain quantities polynomially, but an exact polynomial is irrelevant. 
We thus say that a value $n$ is \emph{polynomially bounded} in $n_1, \ldots, n_k$
if there exists a nondecreasing polynomial $P: \N^k \to \N$ such that $n \leq P(n_1, \ldots, n_k)$
for all $n_1, \ldots, n_k \in \N$. 
We also write $n \leq \poly(n_1, \ldots, n_k)$.

\para{Linear path schemes}
A $d$-dimensional \emph{linear path scheme} ($d$-\LPS in short, or \LPS if dimension is irrelevant) 
is a sequential \vass where every component is either trivial (a singleton) or a simple cycle,
i.e., a \vass whose control graph is a simple path with disjoint simple cycles attached to some states of the path.
We write down \LPS in the following form 
\[
\alpha_0\beta_1^* \alpha_1 \cdots \alpha_{k-1} \beta_{k}^* \alpha_k,
\]
where each $\alpha_i$ and $\beta_i$ is  a fixed sequence of transitions. 
Thus the cycles $\beta_i$ of an \lps may be repeated arbitrarily many times (possibly zero). 
An \LPS is \emph{simple} (\SLPS) when all $\beta_i$ are single transitions, i.e.,
each component is either trivial or a single self-loop transition.
When considering the reachability relation in a $d$-\LPS, we often implicitly take the first and the last state
of the \LPS as the source and target state, respectively, and consider the reachability $w\tran w'$ between
vectors $w,w'\in\N^d$ only.



\begin{lemma}\label{lem:two-slps}
For every \dvass $V$ and two its states $q$, $q'$,
there is a finite set $\Gamma$ of \dslps of size polynomially bounded in $\size(V)$,
such that for all $w,w'\in\N^2$,
\begin{align*} 
 \ \ q(w) \tran q'(w') \text{ in } V \ \iff \ w \tran w' \text{ in some  \dslps in } \Gamma.
\end{align*}
\end{lemma}

\begin{proof}
The claim follows by combination of
Theorem 3.1 from~\cite{DBLP:journals/jacm/BlondinEFGHLMT21}, due to which we get
a finite set $\Gamma$ of \dlps of size polynomially bounded in $\size(V)$ that satisfies the claim,
with Lemma 5.2 from~\cite{DBLP:journals/jacm/BlondinEFGHLMT21} 
(or Theorem 15 from~\cite{DBLP:conf/lics/EnglertLT16}), due to which
we can transform every \dlps $\lambda$ into an \dslps of size polynomially bounded in $\size(\Lambda)$.
\end{proof}

\begin{lemma}\label{lem:slps-sandwich}
\dslps are \sandwich.
\end{lemma}

Before proving Lemma~\ref{lem:slps-sandwich}, we use it together with Lemma~\ref{lem:two-slps}
to prove Lemma~\ref{lem:2vass-sandwich}.
To this aim consider a \dvass $(V, s)$ where $s=p(w)$, and its state $q$, and let $M=\size(V,s)$.
By Lemma~\ref{lem:two-slps} we get a finite set $\Gamma$ of \dslps such that:
\[
\reach_q(V,s) \ = \ \bigcup_{\Lambda \in \Gamma} \reach(\Lambda, w).
\]
Moreover, for some nondecreasing polynomial $F$, every $\Lambda \in \Gamma$ satisfies
$\size(\Lambda) \leq F(M)$.
By Lemma~\ref{lem:slps-sandwich}, there is a nondecreasing polynomial $G$
such that for every $\Lambda \in \Gamma$ and $B\in\N$, 
the set $\reach(\Lambda,w)$
is 
\kanapka {$G(\size(\Lambda,w))$} {$B$}.
Combining the last two statements, we deduce that $\reach_q(V, s)$ is
\kanapka {$G(F(M))$} {$B$} for every $B\in\N$, as required.
Lemma~\ref{lem:2vass-sandwich} is thus proved (once we prove
Lemma~\ref{lem:slps-sandwich}).
\end{appendixproof}

\begin{appendixproof}[Proof of Lemma~\ref{lem:slps-sandwich}]
We generalise finite prefixes $P^{\leq B}$ of $P^*$ as follows.
Let $P = \set{p_1, \ldots, p_m}$. 
For a positive vector $c \in (\Npos)^m$ and $T \in \N$ we define
\[
P^{x \cdot c \leq T} := \setof{\Sigma_{i=1}^m n_i \cdot p_i}{\Sigma_{i=1}^m n_i \cdot c_i \leq T}.
\]
In particular, $P^{\leq T} = P^{x\cdot \vec 1  \leq T}$.
In the sequel, sets of the form
\begin{align} \label{eq:hybrid}
a + P^{x \cdot c \leq T} + Q^*,
\end{align}
for $a\in\N^2$, $c \in (\Npos)^{\card P}$ and $P,Q\subseteqfin \N^2$, 
we call \emph{hybrid} sets.


\begin{lemma}\label{lem:sandwich-raw}
For every \dslps $(\Lambda,s)$,
the set $\reach(\Lambda,s)$ is a finite union of hybrid sets \eqref{eq:hybrid}, 
where $\norm(a),  \norm(c), \norm(P \cup Q)$ are bounded polynomially in $\size(\Lambda,s)$.
\end{lemma}

Before proving Lemma~\ref{lem:sandwich-raw} we use it to prove Lemma~\ref{lem:slps-sandwich}.
We need to argue that there is a nondecreasing polynomial $F$ such that 
for every \dslps $(\Lambda, s)$ and $B\in\N$, the set $\reach(\Lambda,s)$ is
\kanapka {$F(M)$} {$B$}, where $M=\size(\Lambda,s)$.
%
We fix the nondecreasing polynomial $F(x) = R(x) + R^2(x)$, where $R$ is a polynomial witnessing 
Lemma~\ref{lem:sandwich-raw}, and some arbitrary $B \in \N$,
and prove that each hybrid set $H$ \eqref{eq:hybrid} 
of Lemma~\ref{lem:sandwich-raw} is
\kanapka {$F(M)$} {$B$}.
%
%
We distinguish two cases.
If $T \geq R(M) \cdot B$ then, since $\norm(c)\leq R(M)$, we have
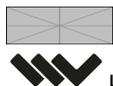
\[
a + (P \cup Q)^{\leq B} \ \subseteq \ H \ \subseteq \ a + (P \cup Q)^*,
\]
and $\norm(a), \norm(P \cup Q) \leq R(M) \leq F(M)$, as required. 
On the other hand, if  $T < R(M) \cdot B$ then $H$ is a finite union
of linear sets of the form $u + Q^*$ for $u \in b + P^{c \cdot x \leq T}$, where
\begin{align*}
\norm(u) \ \leq \ \norm(a) + \norm(P) \cdot T 
\ \leq \ R(M) + R(M)^2 \cdot B \leq  F(M) \cdot B,
\end{align*}
as required. 
As before,  $\norm(Q) \leq R(M) \leq F(M)$.
This completes the proof of Lemma~\ref{lem:slps-sandwich} (once
we prove Lemma~\ref{lem:sandwich-raw}). 
\end{appendixproof}

\begin{appendixproof}[Proof of Lemma~\ref{lem:sandwich-raw}]
The proof occupies the rest of this section.
We rely on an insightful characterisation of paths of \slps \cite[Theorem 4.16]{DBLP:conf/focs/0001CMOSW24},
which we state below using a slightly different terminology.
%
Speaking informally, a \emph{detailing} of an \slps 
$\Lambda = \alpha_0 \beta_1^* \alpha_1 \ldots \alpha_{k-1} \beta_k^* \alpha_k$
is any \slps obtained by fixing exponents of some of the cycles of $\Lambda$.
Formally, a detailing of $\Lambda$ is any $\Lambda'$ obtained by choosing a subset $S \in [1,k]$ and,
for all $i \in S$, by replacing the cycle $\beta_i$ by a path $\beta_i^{n_i}$, for some $n_i\in\N$, 
which becomes an infix of the simple path of $\Lambda'$.
The number of cycles of $\Lambda'$ is thus $k-\card{S}$.
An \dslps is \emph{zigzagging} if the effect of its every cycle $\beta_i$ 
belongs either to the quadrant $\Npos \times (-\Npos)$,
or to the quadrant $(-\Npos) \times \Npos$, and additionally effects of every two consecutive cycles belong 
to different quadrants
(the effects of cycles $\beta_1, \ldots, \beta_k$ thus alternate between quadrants), and the effect of the first cycle belongs to $\Npos \times (-\Npos)$ and the effect of the last cycle belongs to $(-\Npos) \times \Npos$.
Finally, an \slps is \emph{short} if it contains at most three cycles, $k\leq 3$.
For $B\in\N$,
a path 
\[
s_0 \trans{\alpha_0} t_0 \trans{\sigma_1} s_1 \trans{\alpha_1} t_1 \ \cdots \ s_{k-1} \trans{\alpha_{k-1}} t_{k-1} \trans{\sigma_{k}} s_k \trans{\alpha_k} t_k
\]
of an \dslps,
where $\sigma_i \in\beta^*$ for $i\in\setto k$, is called \emph{$B$-close} if all the vectors $x\in \{s_0, t_0, s_1, t_1, \ldots, s_k, t_k\}$,
called \emph{midpoints} below,
are \emph{$B$-close} to some axis, namely either $x\in \setfromto 0 B \times \N$ or 
$x\in \N\times\setfromto 0 B$.

\begin{theorem}[Thm 4.16 in~\cite{DBLP:conf/focs/0001CMOSW24}]\label{thm:2slps-zigzag}
For every \dslps $\Lambda$ there is $B\leq \poly(\size(\Lambda))$
such that for every path $s \tran t$ in $\Lambda$ 
there is a detailing $\Lambda' = \Lambda_1 \Lambda_2 \Lambda_3$ of $\Lambda$  
of $\size(\Lambda')\leq \poly(\size(\Lambda))$ and $u, u' \in \setfromto 0 B \times \N$ such that
\begin{enumerate}
  \item $\Lambda_1$ and $\Lambda_3$ are short, 
  \item $\Lambda_2$ is zigzagging,
  \item there are paths $s \tran u$ in $\Lambda_1$, 
  a $B$-close path $u \tran u'$ in $\Lambda_2$, and 
  $u'\tran t$ in $\Lambda_3$.
\end{enumerate}
\end{theorem}

Fix in the sequel $B$ given by Theorem~\ref{thm:2slps-zigzag}.
There are only finitely many detailings $\Lambda'$ of $\Lambda$ of a bounded size,
only finitely many possible decompositions of $\Lambda'$ into $\Lambda_1$, $\Lambda_2$ and $\Lambda_3$,
and only finitely many values of $u_1, u'_1\in\setfromto 0 B$.
By Theorem~\ref{thm:2slps-zigzag}, vectors $t$ reachable from $s$ in $\Lambda$ are exactly those
reachable from $s$ in some of detailing $\Lambda'$.
Therefore it is enough
to show Lemma~\ref{lem:sandwich-raw} for the set of vectors $t\in\N^2$ reachable by paths as in point 3 above,
in a fixed \slps $(\Lambda',s)$, where $\Lambda' = \Lambda_1 \Lambda_2 \Lambda_3$ satisfies points 1, 2 above,
and where $u_1 = b$ and $u'_1 = b'$ for some fixed $b,b'\in\setfromto 0 B$.

In a path $u\tran u'$, every second midpoint is $B$-close to one axis,
say $x\in \setfromto 0 B \times \N$,
while the remaining midpoints are $B$-close to the other one.
We relax this requirement slightly, by dropping the latter condition.
A path $u \tran u'$  in the zigzagging \slps $\Lambda_2$
is \emph{$B$-vertical-close} if $u$, $u'$ and every second midpoint $x$
are $B$-close to the vertical axis, namely $x \in \setfromto 0 B \times \N$
(thus the remaining endpoints on the path do not have to be $B$-close to the horizontal axis).
In order to have more flexibility in the proof of Lemma~\ref{lem:sandwich-raw},
in the sequel we consider those paths in $\Lambda'$, as in point 3 in Theorem~\ref{thm:2slps-zigzag}, where the infix
$u\tran u'$ is $B$-vertical-close but not necessarily $B$-close.
Notice that by relaxing the condition to $B$-vertical-closeness
we enlarge the set of considered paths, but do not enlarge the set of reachable points,
as every $t$ such that $s \tran t$ is already reachable by the paths with the middle path being even $B$-close.
We denote by $\reach(\Lambda', s)$ the set of all vectors $t\in\N^2$ reachable by such a path $s\tran t$ with
the middle part being $B$-vertical-close.

We formulate below Claims \ref{cl:prefix}, \ref{cl:infix} and \ref{cl:suffix}
(taking care of a prefix, infix, and suffix, respectively, of a path $s\tran t$), 
show how they imply Lemma~\ref{lem:sandwich-raw}, and finally proceed with the proofs of the three claims.
To this aim we introduce some more notation.
Given a start $a \in \N$, a difference $r \in \N$, and a bound $T \in \N_\infty = \N\cup\{\infty\}$,
the set $a + \{r\}^{\leq T}$ is called \emph{$(a, r, T)$-arithmetic}. 
We omit brackets and write $a + r^{\leq T}$.
In particular, $a + r^{\leq 0} = \{a\}$.
A \dslps  $\alpha_1 \beta_1^* \alpha_2 \beta_2^*$ is \emph{one-turn} if 
$\eff(\beta_1) \in \Npos  \times -\Npos$ and $\eff(\beta_2) \in -\Npos  \times \Npos$
(it is thus a special case of short zigzagging \dslps).
For a set $S\subseteq \N^2$, we use the notation $\reach(\Lambda, S) = \bigcup_{s\in S} \reach(\Lambda, s)$.

In Claims \ref{cl:prefix}, \ref{cl:infix} and \ref{cl:suffix}, we focus on source/target vectors
in $\setfromto 0 B \times \N$ and, intuitively speaking, on arithmetic subsets of each 'line' $\{b\}\times \N$.
First, Claim \ref{cl:prefix} states that the reachability set of a short \dslps, intersected with each line,
is a finite union of arithmetic sets.
Second, Claim \ref{cl:infix} states that the reachability set of a one-turn \dslps from 
an arithmetic set inside a line, intersected with another line, is a finite union of arithmetic sets.
Importantly, the starting point grows only additively, by a polynomially bounded amount, 
as we will apply Claim \ref{cl:infix} $\OO(k)$ times.
Finally, Claim \ref{cl:suffix} states that the reachability set of a short \dslps, from an arithmetic set
inside a line, is a finite union of hybrid sets.
All quantities in the claims are bounded polynomially.

\begin{claim}\label{cl:prefix}
For every short \dslps $(\Lambda,s)$ 
and $u_1 \in [0,B]$,
the set $\{u_2 \mid (u_1, u_2) \in \reach(\Lambda,s)\}$ is a finite union of $(a, r, T)$-arithmetic sets, where
$a \leq \poly(B,M)$, $r \leq \poly(M)$, and $M=\size(\Lambda,s)$.
\end{claim}
\begin{claim}\label{cl:infix}
Let $\Lambda$ be a one-turn \dslps,
and $S_1 = a + r^{\leq K}$ for some $a, r, K \in \N$.
Let $u_1, v_1 \in [0, B]$.
The set $R(S_1) := \{v_2 \mid \exists_{u_2 \in S_1} (u_1, u_2) \tran (v_1, v_2) \text{ in } \Lambda\}$
is a finite union of $(a', r', T')$-arithmetic sets with
$a' \leq a + \poly(B,M,r)$ and $r' \leq \max(\poly(M),r)$, where $M=\size(\Lambda)$.
\end{claim}

%
%
%
\begin{claim}\label{cl:suffix}
For every short \dslps $\Lambda$ and $u= (u_1, u_2), p = (p_1, p_2)\in\N^2$, 
the set $\reach(\Lambda, u + \set{p}^{\leq T})$
is a finite union of hybrid sets \eqref{eq:hybrid},
where $\norm(a), \norm(c), \norm(P), \norm(Q) \leq \poly(\size(\Lambda), \norm(u),\norm(p))$.
\end{claim}

We use the three  claims to derive Lemma~\ref{lem:sandwich-raw}.
As said above,
we consider a fixed \slps $\Lambda' = \Lambda_1 \Lambda_2 \Lambda_3$ and source $s\in\N^2$,
and focus on the set $\reach(\Lambda', s)$ of vectors $t\in\N^2$ such that
there are paths 
\begin{align} \label{eq:paths}
\text{$s \tran u$ in $\Lambda_1$,  \ \  
a $B$-vertically close path $u \tran u'$ in $\Lambda_2$, \ \ 
and   $u'\tran t$ in $\Lambda_3$,}
\end{align}
for some $u,u'\in \setfromto 0 B \times \N$, where 
$u_1 = b$ and $u'_1 = b'$ are fixed.
Let $M':=\size(\Lambda',s) \leq \poly(M)$.
%
%
First, by Claim~\ref{cl:prefix}, the set
$\{u_2 \mid (u_1, u_2) \in \reach(\Lambda_1,s)\}$ is a finite union of $(a, r, T)$-arithmetic sets,
where $a, r$ are bounded polynomially in $M'$ and $B$.
Second,
a path $u\tran u'$ is a concatenation of $\ell\leq \size(\Lambda_2) \leq \poly(M')$
paths of one-turn \slps. 
By $\ell$-fold application of Claim~\ref{cl:infix}, the set
$\{u'_2 \mid (u'_1, u'_2) \in \reach(\Lambda_1 \Lambda_2,s)\}$
is a finite union of arithmetic sets
$a' + (r')^{\leq T'}$, where
$a', r' \leq \poly(\size(\Lambda_2)) \leq \poly(M')$.
Indeed, the bound on $a'$ comes from $\ell$-fold addition of values bounded by 
$\poly(B, \poly(M'), r)\leq \poly(B, \poly(M'), \poly(M'))$,
itself bounded by $\poly(M')$:
\begin{equation}\label{eq:aprim}
a' \leq a + \ell \cdot \poly(M') \leq a + \poly(M') \leq \poly(M').
\end{equation}
Finally, by Claim~\ref{cl:suffix} the set $\reach(\Lambda',s)$ is a finite union of hybrid sets \eqref{eq:hybrid},
where
$
\norm(a)$, $\norm(c)$, $\norm(P \cup Q) \leq \poly(M', B) \leq \poly(M).
$ 
We conclude the proof of Lemma~\ref{lem:sandwich-raw}, keeping in mind that
it still remains to demonstrate Claims~\ref{cl:prefix},~\ref{cl:infix},~and~\ref{cl:suffix}.

\medskip

Here is  a corollary of Lemma~\ref{lem:taming}, useful in the proofs of the three claims:
\begin{lemma}\label{lem:solutions}
Consider a system $A\cdot x = b$ of $m$ Diophantine linear equations with $n$ unknowns,
where absolute values of coefficients are bounded by $N$.
Then, the set of solutions is of a form $U+P^*$, where
$\norm(U \cup P) \leq \poly(nN)^{\poly(n,m)}$.
\end{lemma}

\begin{proof}
The solution set is of the form $U+P^*$,
where $U$ is the set of pointwise minimal nonnegative integer solutions of the system,
and $P$ is the set of pointwise minimal nonnegative integer solutions of its homogeneous version
$A\cdot x = \vec 0$.
By Lemma~\ref{lem:taming} each element of $U \cup P$ has norm at most $M= \Oo(nN)^m$.
Therefore, the number of different solutions is at most $(M+1)^n$, which implies 
$\norm(U \cup P) \leq \poly(nN)^{\poly(n,m)}$.
\end{proof}
In the sequel we apply Lemma~\ref{lem:solutions} in case when $n$ and $m$ are constants,
in which case Lemma \ref{lem:solutions} yields the bound $\norm(U+P) \leq \poly(N)$.

\begin{proof}[Proof of Claim~\ref{cl:prefix}]
For $d \in \setto 2$ and
a path $\gamma$, let $\eff_j(\gamma)$ denote the $j$-th coordinate of $\eff(\gamma)$, and
let $\drop_j(\gamma)\in\N$ be the maximal value of $-\eff_j(\delta)$, where $\delta$ ranges over prefixes of $\gamma$,
that is  the maximal amount by which the $j$-th coordinate can be decreased along $\gamma$.

\Wlog we assume that $\Lambda$ has exactly three loops, 
$\Lambda = \alpha_1\beta_1^{*}\alpha_2\beta_2^{*}\alpha_3\beta_3^{*}\alpha_4$. 
We describe paths $s = (s_1, s_2) \tran (u_1, u_2)$ in $\Lambda$,
\begin{align*}
(s_1, s_2) = \ & 
(a^1_1, a^1_2) \trans{\alpha_1} (b^1_1, b^1_2) \trans{\beta_1^{n_1}}
(a^2_1, a^2_2) \trans{\alpha_2} (b^2_1, b^2_2) \trans{\beta_2^{n_2}} \\
& (a^3_1, a^3_2) \trans{\alpha_3} (b^3_1, b^3_2) \trans{\beta_3^{n_3}}
(a^4_1, a^4_2) \trans{\alpha_4} (b^4_1, b^4_2) = (u_1, u_2),
\end{align*}
by the following system of linear Diophantine inequalities, with unknowns 
$a^{i}_1, a^{i}_2, b^{i}_1, b^{i}_2$, for $i\in\setto 4$, and $n_i$, for $i\in\setto 3$,
ensuring that the effects of $\alpha_i$ and $\beta_i^{n_i}$ are respected, 
and that all points along $\alpha_i$ remain nonnegative 
:
%
%
%
%
\begin{align*}
a^{i}_j + \eff_j(\alpha_i) \ = \ & \ b^{i}_j  & a^1_j \ = \ & \ s_j \\
b^{i}_j + n_i \cdot \eff_j(\beta_i) \ = \ & \ a^{i+1}_j & b^4_1 \ = \ & \ u_1 \\
a^i_j - \drop_j(\alpha_i) \ \geq \ & \ 0 
\end{align*}
Notice that each $\beta_i$ is a single transition, so nonnegativity of $(b_i^1, b_i^2)$ and $(a_{i+1}^1, a_{i+1}^2)$
implies that all the vectors along  $\beta_i^{n_i}$ are also nonnegative. Therefore, we do not need to
add an analog of $a^i_j - \drop_j(\alpha_i) \geq 0$ for $\beta_i$ to the above system.
%
By adding dummy variables we change inequalities into equations,
thus obtaining a system $\U$ of linear Diophantine equations.
All the coefficients in $\U$ are bounded by $\max(B, M)$, 
and all coefficients in its homogeneous version are bounded by $M$.
Therefore, by Lemma~\ref{lem:solutions} the solution set of $\U$ is $U+P^*$,
where $\norm(U) \leq \poly(B,M)$ and $\norm(P) \leq \poly(M)$.
By projecting the solution set to the variable $b^4_2$, we deduce that
the set $S:=\{u_2 \mid (u_1, u_2) \in \reach(\Lambda,s)\}$ is a finite union of 
linear sets $a+X^*$, where $a \leq P_1(B,M)$ and $X \subseteq [0,P_2(M)]$,
for some nondecreasing polynomials $P_1, P_2$.
%
%
%
\begin{claim}\label{cl:linear-1dim}
For every $a,b \in \N$ and $B \subseteq [1,b]$, the linear set $a+B^*$ is a finite union
of arithmetic sets $c + d^*$, where $c \leq a+b^3$ and $d \leq b$.
\end{claim}
\begin{proof}
Let $B = \set{b_1, \ldots, b_m}$.
Let $n \in a + B^*$, and let $(k_1, \ldots, k_m) \in \N^m$ be the lexicographically smallest vector such that
$n = a + \sum_{i=1}^m k_i \cdot b_i$.
We observe that $k_i < b$ for all $i\in\setto {m-1}$ since,
supposing $k_i \geq b$ for
$i < m$, we would get a lexicographically smaller  vector
\[
(k_1, \ldots, k_{i-1}, k_i - b_m, k_{i+1}, \ldots, k_m+ b_i) \in \N^m
\]
with the same property. 
Thus $n \in a + r + b_m^*$,
where $r = \sum_{i<m} k_i \cdot b_i \leq b^3$. 
As $n\in a+ B^*$ was chosen arbitrarily, we deduce
\[
a+B^* \ = \bigcup_{c \leq a+b^3, c \in a+B^*} c + b_m^*,
\]
which concludes the proof of Claim~\ref{cl:linear-1dim}.
\end{proof}
By Claim~\ref{cl:linear-1dim}, the set $S$ is a finite union of arithmetic sets  $a+x^*$, where 
$a \leq P_1(B,M) + P_2(M)^3$ and $x \leq P_2(M)$,
which concludes the proof of Claim~\ref{cl:prefix}.
Notice that we actually get $T = \infty$ or $T=0$ in all the arithmetic sets, but that is not needed for our considerations.
\end{proof}


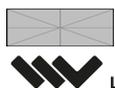
\begin{proof}[Proof of Claim~\ref{cl:infix}]
Fix a  one-turn \slps $\Lambda=\alpha_1\beta_1^*\alpha_2\beta_2^*$, and
let $M:=\size(\Lambda)$ and $S_2 := R(S_1)$.
Our goal is to show that the set $S_2 = \bigcup_{c\in S_1} R(\{c\})$
is a finite union
of $(a', r', T')$-arithmetic sequences for $a' \leq a + \poly(B, M, r)$ and $r' \leq \poly(M)$.
We write $R(c)$ instead of $R(\{c\})$.

Let $\eff(\beta_1) = (x_1, -y_1)$ and $\eff(\beta_2) = (-x_2, y_2)$, for some $x_1, x_2, y_1, y_2 \in \Npos$.
Every path $\rho = \alpha_1\beta_1^{n_1}\alpha_2\beta_2^{n_2}$ is determined by 
a pair $(n_1,n_2) \in \N^2$,
and hence when $(u_1, u_2) \tran (v_1, v_2)$, we may also write $(u_1, u_2) \trans {(n_1,n_2)} (v_1, v_2)$
for $(n_1, n_2)\in\N^2$, or 
$(u_1, u_2) \trans {(n_1,n_2)}$ when the target vector is not relevant.
Whenever this happens, we necessarily have the equality
$
u_1 + \eff_1(\alpha_1 \alpha_2) + n_1 x_1 - n_2 x_2 = v_1
$
of effects on the first coordinate,
which we transform to an equation with two unknowns $n_1,  n_2$:
\begin{equation}\label{eq:xy}
n_1 x_1 - n_2 x_2 = v_1 - u_1 - \eff_1(\alpha_1 \alpha_2).
\end{equation}
The set of solutions $(n_1, n_2)$ of~\eqref{eq:xy} is of the form $w+p^*$, 
where $w = (w_1, w_2) \in \N^2$ is the minimal solution of~\eqref{eq:xy}
and $p = (p_1, p_2) \in \N^2$ is the minimal solution of its uniform version, $n_1 x_1 - n_2 x_2 = 0$.
By Lemma~\ref{lem:solutions} we get $\norm(w) \leq \poly(B, M)$. One can easily observe that $\norm(p) \leq 2M$
($(x_1, x_2) = (n_2, n_1)$ is a solution).
Let $\eff_2(w) = -w_1 y_1 + w_2 y_2$ and $\eff_2(p) = -p_1 y_1 + p_2 y_2\in \N$ be 
the effects induced by $w$ and $p$, respectively, on the second coordinate.

\begin{example}\label{ex:infix}
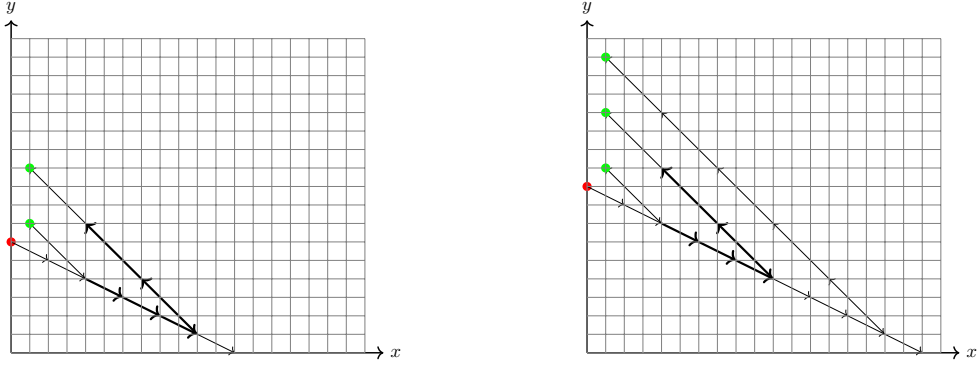
\begin{figure}[t]
\vspace{-2cm}
\begin{minipage}{0.5\textwidth}
\begin{tikzpicture}[scale=0.35]
\scalebox{0.7}{
\draw[->, thick] (0, 0) -- (20, 0) node[right] {$x$};
\draw[->, thick] (0, 0) -- (0, 18) node[above] {$y$};

\fill[red] (0,6) circle (7pt);

\draw[->] (0,6) -> (2,5);
\draw[->] (2,5) -> (4,4);
\draw[very thick,->] (4,4) -> (6,3);
\draw[very thick,->] (6,3) -> (8,2);
\draw[very thick,->] (8,2) -> (10,1);
\draw[->] (10,1) -> (12,0);

\draw[->] (4,4) -> (1,7);
\fill[green] (1,7) circle (7pt);

\draw[very thick,->] (10,1) -> (7,4);
\draw[very thick,->] (7,4) -> (4,7);
\draw[->] (4,7) -> (1,10);
\fill[green] (1,10) circle (7pt);

\draw[step=1, gray, thin] (0, 0) grid (19, 17);
}
\end{tikzpicture}
\end{minipage}
\begin{minipage}{0.5\textwidth}
\begin{tikzpicture}[scale=0.35]
\scalebox{0.7}{
\draw[->, thick] (0, 0) -- (20, 0) node[right] {$x$};
\draw[->, thick] (0, 0) -- (0, 18) node[above] {$y$};

\fill[red] (0,9) circle (7pt);

\draw[->] (0,9) -> (2,8);
\draw[->] (2,8) -> (4,7);
\draw[very thick,->] (4,7) -> (6,6);
\draw[very thick,->] (6,6) -> (8,5);
\draw[very thick,->] (8,5) -> (10,4);
\draw[->] (10,4) -> (12,3);
\draw[->] (12,3) -> (14,2);
\draw[->] (14,2) -> (16,1);
\draw[->] (16,1) -> (18,0);

\draw[->] (4,7) -> (1,10);
\fill[green] (1,10) circle (7pt);

\draw[very thick,->] (10,4) -> (7,7);
\draw[very thick,->] (7,7) -> (4,10);
\draw[->] (4,10) -> (1,13);
\fill[green] (1,13) circle (7pt);

\draw[->] (16,1) -> (13,4);
\draw[->] (13,4) -> (10,7);
\draw[->] (10,7) -> (7,10);
\draw[->] (7,10) -> (4,13);
\draw[->] (4,13) -> (1,16);
\fill[green] (1,16) circle (7pt);

\draw[step=1, gray, thin] (0, 0) grid (19, 17);
}
\end{tikzpicture}
\end{minipage}

\caption{Left: $u_2 = 6$, $S_2 = \set{7,10}$. Right: $u_2 = 9$, $S_2 = \set{10,13,16}$.
Thick vectors add up to $p$.}
\label{fig:infix}
\end{figure}

Let $\alpha_1$ and $\alpha_2$ be empty sequences, $\eff(\beta_1) = (2, -1)$ and $\eff(\beta_2) = (-3, 3)$.
Let $u_1 = 0$, $v_1 = 1$. 
The set of solutions $(n_1, n_2)$ of~\eqref{eq:xy} is of the form $w+p^*$, where
$w = (2,1)$ and $p = (3,2)$, and
$\eff_2(p) = 3 \cdot (-1) + 2 \cdot 3 = 3$.
Figure~\ref{fig:infix} shows $R(u_2) = \set{7,10}$ when $u_2 = 6$, and $R(u_2)  = \set{10,13,16}$ when $u_2 = 9$.
In the latter case, elements of $R(u_2)$ correspond to the solutions $w$, $w+p$ and $w+2p$ of~\eqref{eq:xy},
i.e., to $w + kp$ for $k$ in an interval $[0,2]$.
According to Claim~\ref{cl:interval}, this is true in general.
\end{example}


\begin{claim}\label{cl:interval}
For each $c \in S_1$ there exists an interval $I_{c} = [k_1, k_2]$, where $k_1\in \N$ and $k_2 \in \N_{\infty}$,
such that
$R(c) = \setof{\eff_2(w) + k \cdot \eff_2(p)}{k\in I_{c}}$.
\end{claim}

\begin{proof}
It is sufficient to show that whenever $(u_1, c) \trans{w + k_1 \cdot p}$ and
$(u_1, c) \trans{w + k_2 \cdot p}$ for some
$k_1 < k_2 \in \N$,
then
$(u_1, c) \trans{w + k \cdot p}$ for all $k \in [k_1, k_2]$.
Fix $k \in [k_1, k_2]$.
Every point $x$ on the (possibly $\Z$-)path $(u_1, c) \trans{w + k \cdot p}$ is actually on a straight line between some two points,
one on the path $(u_1, c) \trans{w + k_1 \cdot p}$, and the other on the path $(u_1, c) \trans{w + k_2 \cdot p}$.
In consequence $x$, being a weighted average of the two points in $\N^2$, necessarily belongs to $\N^2$. 
Therefore $(u_1, c) \trans{w + k \cdot p}$ is a path.
\end{proof}

We notice that $\eff_2(p)$ can be negative, but this is irrelevant for our arguments.
By Claim~\ref{cl:interval}, for each $c \in S_1$ the set $R(c)$ is
is an arithmetic sequence of difference $\essdvass r := \absv{\eff_2(p)}$ and length equal to the cardinality
of the interval $I_{c}$.
Let $\spn(c)\in\N_\infty$ be the difference between the supremum of $R(c)$ and the minimal element in $R(c)$.
We say that $\spn(c) = \infty$ if $R(c)$ is infinite.
We split the proof into cases, depending on
whether $\essdvass r$ divides the difference $r$ of the sequence $S_1$, or not.
Additionally we have a case when $\essdvass r = 0$, in other cases we silently assume that $\essdvass r \neq 0$.

\para{Case I: $r$ is divisible by $\essdvass r$}

Therefore, if $\spn(c) \geq r$ then the sequence $R(c)$ actually touches the sequence $R(c+r)$,
i.e., their union is a larger arithmetic sequence of difference $\essdvass r$:

\begin{claim}\label{cl:merged}
If $R(c+r^*) \neq \emptyset$ and $c \geq D:= 3(M+r) \cdot M^2 + M \cdot \norm(w)$,
then $R(c+r^{\leq T}) = b+(\essdvass r)^{\leq T'}$ for some $b \leq c + \eff_2(w) + M \cdot \eff_2(p)$ and $T' \in \N_\infty$.
\end{claim}

\begin{proof}
We first show that $\spn(c) \geq r$.
Suppose $R(c+r^*) \neq \emptyset$ and $c \geq D$.
Due to the first assumption,  for some $n\in\N$ there is a path $(u_1, c+nr) \trans{(n_1, n_2)}$
of the form
\begin{align} \label{eq:Mp}
(u_1, c+nr) \trans{\alpha_1} (\essdvass{x}_1, \essdvass{y}_1) \trans{\beta_1^{n_1}} (\essdvass{x}_2, \essdvass{y}_2) \trans{\alpha_2} (\essdvass{x}_3, \essdvass{y}_3) \trans{\beta_2^{n_2}} (v_1, v_2).
\end{align}
In particular,  $u_1 \geq \drop_1(\alpha_1)$
and therefore $\essdvass x_1\geq 0$.
It is enough to take $(n_1, n_2) := w + Mp$ in \eqref{eq:Mp}.
Then $n_1 \geq M$, so $\essdvass x_2 \geq \essdvass x_1 + M \geq M \geq \drop_1(\alpha_2)$.
Therefore $\essdvass x_3\geq 0$.
Now we show that $c$ is large enough such that $(u_1, c) \trans{w+Mp}$ is nonnegative on the second coordinate as well.
As $\norm(p) \leq 2M$ then $n_1 + n_2 \leq M \cdot \norm(p) + \norm(w) \leq 2M^2 + \norm(w)$. Therefore
$\beta_1^{n_1}$ and $\beta_2^{n_2}$ can in total decrease the second coordinate by at most $M \cdot (2M^2 + \norm(w))$.
As $\alpha_1 + \alpha_2$ can in total decrease the second coordinate by at most $M$ and
$D \geq M + M \cdot (2M^2 + \norm(w))$ we conclude that indeed the path $(u_1, c) \trans{w+Mp}$ is valid.
For the same reasons, for every $m\in [M, M+r]$ there is a path $(u_1, c) \trans{w+mp}$, which guarantees $\spn(c) \geq r$. 

Now we use that fact that $\spn(c) \geq r$.
By monotonicity of \vass, if $R(c) = b+(\essdvass r)^{\leq {T'}}$
then $R(c+r)$ necessarily includes $R(c) + r= b+r+(\essdvass r)^{\leq {T'}}$.
Since $\spn(c) \geq r$, we have $b+r \in R(c)$, but also $b+r \in R(c+r)$,
and therefore the union $R(c) \cup R(c + r)$ forms one arithmetic sequence 
$b+(\essdvass r)^{\leq T'}$, for some $b\in\N$ and $T'$.
The similar reasoning applies to any finite union, namely to $R(c+r^{\leq m})$ for $m\in\N$.
In consequence, for every $T\in \N_\infty$ we have $R(c+r^{\leq T}) = b + (\essdvass r)^{\leq T'}$,
for some $b\in\N$ and $T'\in \N_\infty$,
and since $c+ \eff_2(w) + M \cdot \eff_2(p) \in R(c)$ we get the inequality $b \leq c + \eff_2(w) + M \cdot \eff_2(p)$, as required.
\end{proof}

We are ready for concluding Case I.
As $\norm(w) \leq \poly(B, M)$ we have $D \leq \poly(B, M, r)$.
We  partition $S_1 = a + r^{\leq T}$ into two subsets: $S'_1 = S_1 \cap [0,D)$
and $S''_1 = S_1 \cap [D,\infty)$, both being arithmetic sequences of difference $r$,
and consider $S'_1$ and $S''_1$ separately.
 
Concerning $S'_2:= R(S'_1)$,
as $\max(S'_1) \leq D$, all
elements of $S'_2$ are upper-bounded by a polynomial in $M$ and $r$, namely
$\max(S'_2) \leq M^2 \cdot (2M + D)$. 
Thus $S'_2$ can be seen as a finite sum of singletons, each of which
being an $(a', r', T')$-arithmetic sequence with $a' \leq M^2 \cdot (2M + D) \leq \poly(B,M,r)$,
$r' = 1$ and $T' = 0$. 
Clearly $r' = 1 \leq \poly(M)$, and hence $S'_2$ is of the required form.

Now we consider $S''_2:= R(S''_1)$.
If $S''_2 = \emptyset$ we are done.
Otherwise, let $c:= \min(S''_1)$.
Thus $S''_1 = c+r^{\leq T}$ for some $T \in \N_\infty$, 
and $D\leq c \leq \max(a, D+r)$.
By Claim~\ref{cl:merged} we deduce that
$S''_2 = b + (\essdvass r)^{\leq T'}$ for some 
$b \leq c + \eff_2(w) + M \cdot \eff_2(p)$ and $T' \in \N_\infty'$. 
As $\eff_2(w) \leq \poly(B,M)$, $\eff_2(p) \leq \poly(M)$ and $c \leq a + \poly(B,M, r)$ we get $b \leq a + \poly(B, M, r)$.
We also have $\essdvass r \leq \poly(M)$, and hence
$S''_2$ is of the required form.

\para{Case II: $r$ is not divisible by $\essdvass r$}

In that case we split $S_1 = a + r^{\leq T}$ into several
arithmetic sequences of difference $r \cdot \essdvass r$, namely into sequences of a form 
$(a + m \cdot r) + (r \cdot \essdvass r)^{\leq T'}$, where $m < \essdvass r$, 
and apply the above reasoning to each of this sequences separately. 
As $\essdvass r \leq \poly(M)$
we get also a finite set of arithmetic sequences with the base bounded by $a + \poly(B, M,r)$ and difference bounded by $\poly(M)$,
as required.

\para{Case III: $\essdvass r = 0$}
\Wlog we assume $R(S_1) \neq \emptyset$. For every $c \in S_1$ we have that either $R(c) = c + \eff_2(w)$ or $R(c) = \emptyset$. By monotonicity of VASS we have that if $R(c) \neq \emptyset$ then $R(c+r) \neq \emptyset$. Let $D := 3(M+r) \cdot M^2 + M \cdot \norm(w)$. Similarly as in the proof of Claim~\ref{cl:merged} we observe that if $c \geq D$ then for some $k \in N$ there is a run $(u_1, c) \trans{w+kp}$. Hence we have $R(S_1) = c + \eff_2(w) + r^{\leq T'}$ for some $T' \in \N_\infty$ and some $c \in S_1$ such that $c \leq \max(a, D+r)$. Therefore $R(S_1) = b + r^{\leq T'}$ for some $b \leq a + \poly(B,M,r)$ as required.

\end{proof}

\begin{proof}[Proof of Claim~\ref{cl:suffix}]
\Wlog we assume $\Lambda = \alpha_1\beta_1^*\alpha_2\beta_2^*\alpha_3\beta_3^*\alpha_4$.
Let $S_1 = b + \set{p}^{\leq T}$  and $S_2 = \reach_{\Lambda}(S_1)$. Recall, that $\eff_j(\gamma)$ denotes the $j$-th coordinate of $\eff(\gamma)$, and
$\drop_j(\gamma)\in\N$ is the maximal value of $-\eff_j(\delta)$, where $\delta$ ranges over prefixes of $\gamma$,
that is  the maximal amount by which the $j$-th coordinate can be decreased along $\gamma$.

Similarly as in the proof of Claim~\ref{cl:prefix} we describe paths $s = (s_1,s_2) \tran (t_1, t_2) = t$ in $\Lambda$,
\begin{align*}
(s_1, s_2) = &
(a_1^1, a_2^1) \trans{\alpha_1} (b_1^1, b_2^1) \trans{\beta_1^{n_1}}
(a_1^2, a_2^2) \trans{\alpha_2} (b_1^2, b_2^2) \trans{\beta_2^{n_2}} \\
& (a_1^3, a_2^3) \trans{\alpha_3} (b_1^3, b_2^3) \trans{\beta_3^{n_3}}
(a_1^4, a_2^4) \trans{\alpha_4} (b_1^4, b_2^4) = (t_1, t_2).
\end{align*}
Notice that $(s_1, s_2) = (a_1^1, a_2^1) \in S_1$, so $(a_1^1, a_2^1) = u + p^n$ for some $n \in \N$. The following system of linear Diophantine inequalities $\U$ with unknowns $a_1^i, a_2^i, b_1^i, b_2^i$ for $i \in [1,4]$, and $n_i$, for $i \in [1,3]$, and unknown $n$, ensures that the effects of $\alpha_i$ and $\beta_i^{n_i}$ are respected, and that all points along $\alpha_i$ remain nonnegative and that $s \in S_1$:

\begin{align*}
a^{i}_j + \eff_j(\alpha_i) \ = \ & \ b^{i}_j  & a^1_j \ = \ & \ u_j + n \cdot p_j \\
b^{i}_j + n_i \cdot \eff_j(\beta_i) \ = \ & \ a^{i+1}_j & n \ \leq \ & \ T \label{eq:LPS_bound} \\
a^i_j - \drop_j(\alpha_i) \ \geq \ & \ 0 
\end{align*}
Notice that each $\beta_i$ is a single transition, so nonnegativity of $(b_i^1, b_i^2)$ and $(a_{i+1}^1, a_{i+1}^2)$
implies that all the vectors along  $\beta_i^{n_i}$ are also nonnegative. Therefore, we do not need to
add an analog of $a^i_j - \drop_j(\alpha_i) \geq 0$ for $\beta_i$ to the above system. We first focus on the solutions of $\U$ without the equation $n \leq T$, 
and inequalities $a^i_j - \drop_j(\alpha_i) \geq 0$ transformed into equations with dummy variables on the right, similarly as in the proof of Claim~\ref{cl:prefix}. Let us call such system $\U'$. All the coefficients of $\U'$ are bounded by $\max(M, \norm(u), \norm(p))$. By Lemma~\ref{lem:solutions} set of solutions of $\U'$ can be described as $L(U, V)$ for $\norm(U \cup V) \leq \poly(M, \norm(u), \norm(p))$.


Now we have to care about the last inequality, namely $n \leq T$. If for some $a \in U$ we have $n > T$ then we can remove it from $U$. Let $U'$ be the set $U$ without the removed elements.
As the set $U'$ is finite it is enough to prove the conclusion of Claim~\ref{cl:suffix} separately for each $a \in U'$. Fix $a \in U'$.
We have $\norm(a), \norm(V) \leq \poly(M, \norm(u), \norm(p))$.
It is enough to prove that elements of the set $L(a, V)$ that additionally satisfy $n \leq T$, projected to the variables
$b^4_1$ and $b^4_2$ can be described as a finite union of the sets of the form we need.

Let us consider all the elements of set $V$. Let $Q \subseteq V$ be the set of these elements $v \in V$,
for which $v[n] = 0$ (that means that unknown $n$ is equal to $0$ in elements $v$),
while $P = V \setminus Q$ be the set of the other elements $v \in V$,
so that for which $v[n] > 0$.
Notice that using elements in $Q$ does not influence satisfying $n \leq T$, therefore they
can be used unbounded number of times.
Let $P = \set{p_1, \ldots, p_\ell}$ and let $p_i[n] = c_i$. For each $x \in L(a, V)$ we have $x = a + q + \Sigma_{i=1}^{\ell}k_i \cdot p_i$ where $q \in Q^*$. Hence, in order to satisfy $n \leq T$ we have to satisfy $\innprod c k \leq (T - a[n])$, where $c = (c_1, \ldots, c_\ell) \in \N_{>0}^\ell$ and $k = (k_1, \ldots, k_\ell) \in \N^\ell$.
As $c \in \N_{>0}^\ell$, if $T - a[n] < 0$ then the set of solutions is empty. Otherwise, the set of solutions can be represented as $a+P^{c \cdot x \leq T - a[n]}+Q^*$. Recall that $\norm(a), \norm(P \cup Q) \leq \poly(M, \norm(u), \norm(p))$, and additionally $c \in \N_{>0}^{\ell}$. Additionally $\norm(c) \leq \ell \cdot \norm(P)$, where
$\ell$ is the number of elements in $P$, thus $\ell \leq (\norm(P) + 1)^k$, where $k$ is the number of unknowns in $\mathcal U$
(so $k$ is a constant).
In consequence $\norm(c) \leq \poly(M, \norm(u), \norm(p))$. Summarising, the projection
of $a+P^{c \cdot x \leq T - a[n]}+Q^*$ into variables $b^4_1$ and $b^4_2$ is of the required form.
\end{proof}
Claims \ref{cl:prefix}, \ref{cl:infix} and \ref{cl:suffix} are thus shown, and hence so is Lemma~\ref{lem:sandwich-raw}.
\end{appendixproof}

In the proof of Lemma~\ref{lem:main} we actually need polynomial approximability not only for \dvass, but also for
its generalisation, \geomvass. It is stated below and shown in the Appendix using Lemma~\ref{lem:2vass-sandwich}.

\begin{lemma} \label{lem:geom-sandwich}
\Geomvass are \sandwich.
\end{lemma}
\begin{appendixproof}[Proof of Lemma~\ref{lem:geom-sandwich}]
Fix an arbitrary \geomvass $(V, s)$ and let $M = \size(V,s)$.
Norms of vectors generating $\cone V$ --- i.e., effects of simple cycles --- are at most $M$, 
as no transition repeats along a simple cycle.
The effect $\delta \in \Z^3$  of each simple cycle
satisfies $\innprod a \delta = 0$, where $a\in\Z^3$ is a vector orthogonal to $\Lin V$,
or equivalently, orthogonal to some two effects of simple cycles.
The vector $a$ is thus an integer solution of a system of 2 linear equations with 3 unknowns,
where absolute values of coefficients are bounded by $M$.
By Lemma \ref{lem:taming}, there is such an integer solution 
$a=(a_1, a_2, a_3)$ with $\norm(a)\leq D=\OO(M^2)$.

In consequence, on every path $s \tran t$ the value of inner product with $a$ is bounded polynomially
with respect to $M$:
\begin{claim} \label{claim:axx}
Every configuration $q(x)\in \reach(V,s)$  satisfies
$-C \leq \innprod a x \leq C$, where $C = \OO(M\cdot D)$. 
\end{claim}

We rely on the construction of \cite[Lemma 5.1]{Zhang-geom}, which transforms a 
\geomvass $(V,s )$ into a \dvass $(\essdvass V, \essdvass s)$ of size at most $R(M)$
for some polynomial $R$, by dropping on of dimensions of $V$.

\para{Case I: $a$ contains both positive and negative numbers}
\Wlog assume that $a_1, a_2 \geq 0$ and $a_3<0$,
in which case it is the third coordinate which is dropped 
by the construction of \cite[Lemma 5.1]{Zhang-geom}.
States of $\essdvass V$ are of the form $\pair q c$, where $q\in Q$ and $c\in\setfromto {-C} C$, 
plus some further auxiliary states, omitted here.
Due to Claim \ref{claim:axx},
there is a one-to-one correspondence between
reachable configurations in $V$ and reachable configurations in $\essdvass V$:
\[
e = q(x_1, x_2, x_3) \quad \longmapsto \quad
\essdvass e = \pair q c(x_1, x_2), 
\qquad
\text{ where } c = \innprod a x.
\]
The tight correspondence between paths of $V$ and $\essdvass V$,
 given in Claim \ref{claim:V321} below, is essentially 
Lemma 5.1 of \cite{Zhang-geom}:
\begin{claim} \label{claim:V321}
For every configurations $s,u$,
here is a path $s \tran u$ in $\essdvass V$
if, and only if,
there is a path $\essdvass s \tran \essdvass u$ in $\essdvass V$.
\end{claim}
By Lemma \ref{lem:2vass-sandwich}, there is a polynomial $F$ such that 
for every $B\in\N$,
in the \dvass
$(\essdvass V, \essdvass s)$ obtained by the above construction, 
for every its state $\pair q c$, the set 
$\reach_{\pair q c}(\essdvass V, \essdvass s)$ is \kanapka {$F(M')$} {$B$}, where
$M'=\size(\essdvass V, \essdvass s) \leq R(M)$, and hence also
\kanapka {$F(R(M))$} {$B$}.
We claim that for every state $q\in Q$, for every $B\in\N$, 
the set $\reach_q(V,s)$ is \kanapka {$G(F(R(M)))$} {$B$}, for some nondecreasing polynomial $G$.
Indeed, for any $B\in\N$, any ($B$-approximation of) a linear set 
$L = w + P^*\subseteq \N^2$, where $w=(w_1, w_2)$, witnessing that
$\reach_{\pair q c}(\essdvass V,\essdvass s)$  is \kanapka {$F(R(M))$} {$B$}
is transformed to a ($B$-approximation of) linear set $L'$ witnessing that
$\reach_{q}(V,s)$  is \kanapka {$G(F(R(M)))$} {$B$}, as follows.
Take as base the unique vector $w'=(w_1, w_2, w_3)$ such that $a_1 w_1 + a_2 w_2 + a_3 w_3 = c$.
For every period $p=(p_1, p_2) \in P$, take into the set $P'$  the unique vector
$p'=(p_1, p_2, p_3)$ such that $a_1 p_1 + a_2 p_2 + a_3 p_3 = 0$.
Since $a_3>0$, it is guaranteed that $p_3 \geq 0$, and therefore $p'\in\N^3$.
Let the polynomial $G$ bound the blowup of $\norm(b')$ with respect to $\norm(b)$, and 
$\norm(p')$ with respect to $\norm(p)$, for instance
$G(x) = M \cdot x + B$.
The union of all ($B$-approximations of) 
so described sets $L' = w' + (P')^*$, for all $c\in\setfromto {-C} C$, provides the witness that
$\reach_{q}(V,s)$ is \kanapka {$G(F(R(M)))$} {$B$}.

\para{Case II: $a$ is non-negative or non-positive}
\Wlog assume $a\geq \vec 0$ and $a_3 > 0$.
By Claim \ref{claim:axx}, for each $q(x) \in \reach(V,s)$
we thus have $x_3 \leq C$.
We transform $(V,s)$ into $(\essdvass V, \essdvass s)$ 
with states of the form $\pair q c$, where $q\in Q$ and $c\in\setfromto {0} C$, 
by storing the third coordinate in state:
\[
e = q(x_1, x_2, x_3) \quad \longmapsto \quad
\essdvass e = \pair q c(x_1, x_2), 
\qquad
\text{ where } c = x_3.
\]
As above, $\size(\essdvass V, \essdvass s) \leq R(M)$, for a polynomial $R$.
The argument  that for every $B\in\N$,
the set $\reach_{q}(V,s)$ is \kanapka {$G(F(R(M)))$} {$B$}, is similar to Case I (but simpler).
\end{appendixproof}

\section{Proof of Lemma~\ref{lem:main}}\label{sec:mainproof}

In this section we prove Lemma \ref{lem:main}, 
by induction on $k$.
The base of induction, when $k=1$, follows by Lemma \ref{lem:1comp}:
 1-component \tvass are \plb. 
Before engaging in the induction step
we need to generalise wideness, defined up to now for 1-component \tvass only,
to all sequential \tvass.

\para{Sequential cones}

Consider a $k$-component \tvass $V=\ktvass$.
By a \emph{cascade} we mean a tuple of $k$ vectors $(v_1, \ldots, v_k)$  
such that the partial sum $v_1 + \ldots + v_i \in (\Qpos)^3$ for every $i\in\setto k$.
Then the \emph{sequential cone} of $V$, denoted $\seqcone V$, is the set of 
sums of all cascades $(v_1, \ldots, v_k)$ 
whose every $i$th vector $v_i$ belongs to $\cone{V_i}$: 
\[
\seqcone V = \{v_1 + \ldots + v_k \mid (v_1, \ldots, v_k) \in \cone{V_1} \times \ldots \times \cone{V_k} 
\text{ is a cascade}\}.
\]
\vspace{-0.7cm}
\begin{claim} \label{claim:seqcone}
If $\Lin V$ is 3-dimensional then $\seqcone V$ is a finitely generated open cone.
\end{claim}
\begin{appendixproof}[Proof of Claim \ref{claim:seqcone}]
$\seqcone V$ is equivalently definable as the last element $C_k$ of the sequence
of (rational) open cones $C_1, \ldots, C_k$, defined as follows.
We put $C_1 := \cone{V_1}\cap (\Qpos)^3$,
and for $i > 1$ we define inductively:
\[
C_{i} \ := \ \big(C_{i-1}\  + \ \cone{V_{i}}\big) \cap (\Qpos)^3,
\]
where the addition is Minkowski sum $X+Y = \setof{x+y}{x\in X, y\in Y}$.
Then all $C_1, \ldots, C_k$ are finitely generated open cones, as $(\Qpos)^3$ is such
a cone,  
and Minkowski sum and intersection preserve finitely generated open cones.
\end{appendixproof}
We prove the fundamental property:
all reachable configurations are at close distance to the sequential cones.
We focus on \tvass, but actually the same proof works for \vass in any other fixed dimension.
Below, let $d(x,y)$ denote Euclidean distance between $x$ and $y$, and let $d(x,S)$ denote 
the distance between $x$ and a set $S$, that is $d(x,S) = \inf \setof{d(x,y)}{y \in S}$. 
\begin{lemma}\label{lem:not_far_from_cone}
There exists a nondecreasing polynomial $P$ such that 
each reachable configuration $q(w)$
in a forward-diagonal sequential \tvass $(V,s)$,  satisfies
$d(w, \seqcone V) \leq P(\size(V,s))$.
\end{lemma}

\begin{proof}
Let $V$ be a $k$-component sequential \tvass, and $M:=\size(V, s)$.
Let $s=p_1(w)$ and suppose $s \trans{\pi} q(x)$.
\Wlog we assume that $q(x)$ is in the last component $V_k$; indeed, if $q(x)$ is in $V_\ell$ for some
$\ell < k$, we prove that claim for $V'=\ltvass$ and rely on the inclusion
$\seqcone {V'} \subseteq \seqcone V$.
Our aim is to define a polynomial $P$ and a point $y \in \seqcone{V}$ such that $d(x,y) \leq P(M)$.
The path $\pi$ decomposes into $\pi = \pi_1;\, u_1;\, \ldots; \, \pi_{k-1}; \, u_{k-1}; \, \pi_k$ where 
$u_1, u_2, \ldots, u_{k-1}$ are bridge transitions.
Let $p_i$ be the source state of $\pi_i$ and $p'_i$ be the target state of $\pi_i$.
Let $\sigma_i$ be a shortest path from $p'_i$ to $p_i$ 
(there is such a path because they are in the same strongly connected component $V_i$). 
Let $v_i \in \Z^3$ be the effect of $\pi_i;\, \sigma_i$, which is a cycle in $V_i$. 

As $\cone{V_i}$ is an open cone, we are not guaranteed that $v_i \in \cone{V_i}$. 
In order to ensure this property, we add to $v_i$ some small multiples of all the simple cycles in $V_i$.
For each $i \in \setto k$, let $c_i\in\Z^3$ 
be the sum of effects of all simple cycles in $V_i$ and let 
$\eps \in \Q_{>0}$ be a small positive rational such that for all $i \in \setto k$ we have $\norm(\eps \cdot c_i) < 1$. 
Then $v_i + \eps \cdot c_i \in \cone{V_i}$.

By forward-diagonality, the configuration $p_1(w+\vec 1)$ is coverable in $V_1$ from $s$.
Due to the upper bound of Rackoff~\cite[Lemma 3.4]{DBLP:journals/tcs/Rackoff78}, instantiated to the
fixed dimension 3,
for some nondecreasing polynomial $R$ the length of such a covering path $p_1(w) \trans{} p_1(w + \Delta)$,
and the norm of $\Delta$, are both at most $R(M)$.
For some $m \in \N$, the vector $m \cdot \Delta + v_1 + \eps \cdot c_1$ 
belongs to $\cone{V_1}$ and the $k$-tuple
\begin{align} \label{eq:cascade}
(m \cdot \Delta + v_1 + \eps \cdot c_1, \ v_2 + \eps \cdot c_2, \ \ldots, \ v_k + \eps \cdot c_k)
\end{align}
is a cascade.
Therefore the sum $y = m \cdot \Delta + \Sigma_{j=1}^k (v_j + \eps\cdot c_j) \in \seqcone V$.
We argue that is is enough to use polynomially bounded value of $m$.
Since $\pi$ starts in $s=p_1(w)$, each its prefix can drop by at most $\norm(w) \leq M$ on any coordinate.
Thus, for each $j \in \setto k$ we have 
$\eff(\pi_1;\, u_1;\, \ldots; \, \pi_{k-1}; \, u_{j-1}; \, \pi_j) \geq -\vec{M}$.
As $u_i$ are single transitions, we also have
$\eff(u_1) + \ldots + \eff(u_{j-1}) \leq \vec{M}$. 
In consequence, $\eff(\pi_1) + \ldots + \eff(\pi_j) \geq -2\vec{M}$.
Moreover, the sum of norms of effects of all the paths $\sigma_j$ is at most $-M$, 
which implies that 
$v_1 + \ldots + v_{j} \geq -3\vec{M}$. 
Finally,
for each $j \in \setto k$ we have 
$\norm(\eps \cdot (c_1 + \ldots + c_j)) \leq j \leq M$. 
Therefore, setting $m = 4M$ guarantees that
the tuple \eqref{eq:cascade} is a cascade.

Let $C = \eps \cdot (c_1 + \ldots + c_k)$, 
$S = \eff(\sigma_1) + \ldots + \eff(\sigma_k)$ and $U=\eff(u_1) + \ldots + \eff(u_k)$. 
We have $\norm(C), \norm(S), \norm(U) \leq M$, and
$y = x + m\cdot \Delta + S + C - U$.
Therefore $\delta = y - x$ satisfies $\norm(\delta) \leq P(M) := 4M\cdot R(M) + 3M$, and we get the bound
\begin{align*}
d(x, y)^2 \ = \ d(x, x + \delta)^2 \ = \ \innprod \delta \delta \ \leq \ \norm(\delta)^2  
\end{align*}
that implies $d(x,y) \leq \norm(\delta)$, and hence $d(x,y) \leq P(M)$ as required.
\end{proof}
We say that a $k$-component \tvass $V$ is \emph{wide} if $(\Qpos)^3 \subseteq \seqcone V$
or $(\Qpos)^3 \subseteq \seqcone {\rev V}$.
For $k=1$, the definition relaxes the definition of Section \ref{sec:1comp}.

\begin{proof}[Proof of Lemma \ref{lem:main}]
Lemma \ref{lem:trueind}, formulated below, is a refinement of Lemma \ref{lem:main} suitable for 
inductive reasoning.
When proving the lemma, we distinguish 2 cases, depending on whether
a \tvass  is diagonal and wide (we call the \tvass \emph{easy} in this case), 
or not (we call it \emph{non-easy} then), and rely on the following two facts:
%
\begin{lemma} \label{lem:dw}
Easy sequential \tvass are \lb by $P_k(M) = M^{\OO(k)}$, where $k$ is the number of components.
\end{lemma}
%
%
\begin{lemma} \label{lem:ne}
There is a nondecreasing polynomial $H$ such that for every $k > 1$,
if $(k-1)$-component \tvass are \lb by a function $h$ then
non-easy $k$-component \tvass are \lb by the function $H \circ h \circ H \circ h \circ H$.
\end{lemma}
For stating Lemma \ref{lem:trueind},
we define inductively a sequence of polynomials $(h_i)_{i\in\N}$, where
$h_1$ is the polynomial witnessing Lemma \ref{lem:1comp}
(thus 1-component \tvass are \lb by $h_1$), and
\[
h_{j+1} = H \circ h_j \circ H \circ h_j \circ H.
\]
\Wlog we also assume that $h_j$ dominates the polynomial $P_j$ of Lemma \ref{lem:dw}, namely
$P_j(m) \leq h_j(m)$ for every $j,m\geq 1$.

\begin{lemma} \label{lem:trueind}
For every $k\geq 1$, 
$k$-component \tvass are \lb by $h_k$.
\end{lemma}
\begin{proof}
The induction base is given by Lemma \ref{lem:1comp}.
For the induction step, suppose $(k-1)$-component \tvass are \lb by $h_{k-1}$.
Easy $k$-component \tvass are \lb by $h_k$ due to Lemma \ref{lem:dw} and the above-assumed domination, 
without referring to induction assumption,
while non-easy $k$-component \tvass are \lb by $h_k$ due to Lemma \ref{lem:ne}.
\end{proof}
%
%
Let $c\in\N$ be large enough so that $H(m), h_1(m) \leq m^{c}$ for every $m>1$,
and let $C=c^3$.
\begin{claim} \label{claim:kkk}
$h_{k}(m) \leq {m}^{C^{2^{k} -1}}$ for all $m > 1$.
\end{claim}
\begin{appendixproof}[Proof of Claim \ref{claim:kkk}]
The inequality is shown easily by induction on $k$.
When $k=1$, by definition of $c$ we have $h_1(m) \leq m^c\leq m^C$, as required.
Furthermore, assuming
$h_{k-1}(m) \leq {m}^{C^{2^{k-1} -1}}$ for all $m>1$, we get
\[
h_k(m) \ \leq \ {m}^{C\cdot (C^{2^{k-1} -1})^2} 
\ = \ m^{C^{2^k-1}},
\]
as required.
\end{appendixproof}
Lemma \ref{lem:trueind} implies Lemma \ref{lem:main}, 
as the right-hand side of the inequality in Claim \ref{claim:kkk} is bounded by $\kbound {m} {\OO(k)}$.
Lemma \ref{lem:main} is thus proved (once we prove Lemmas \ref{lem:dw} and \ref{lem:ne}).
\end{proof}

The proof of Lemma \ref{lem:dw} generalises Case 1 of the proof of Lemma \ref{lem:1comp}.
The proof of Lemma \ref{lem:ne} makes crucial use of \sandwich sets introduced
in Section \ref{sec:tools}, and builds on 
Lemma \ref{lem:len-eq}, stated below, whose proof 
generalises Cases 2 and 3 of the proof of Lemma \ref{lem:1comp}.

\begin{appendixproof}[Proof of Lemma \ref{lem:dw}]
%
Consider a $k$-component \tvass  $V = \ktvass$,
where $V_i = (Q_i, T_i)$ and $u_i = (p'_i, \delta_i, p_{i+1})$,
together with source and target configurations: 
$s=p_1(w)$ in $V_1$ and $t=p'_k(w')$ in $V_k$.
Let $M = \size(V, s, t) = \size(V) + \norm(s) + \norm(t)$.

Suppose $(V, s, t)$ is diagonal and wide, say $(\Qpos)^3 \subseteq \seqcone{V}$.
We have
$p_1(w) \trans\pi p_1(w+\Delta)$ and
$p'_k(w'+\Delta')\trans{\pi'} p'_k(w')$  for some $\Delta,\Delta'\in(\Npos)^3$, and 
$(\Qpos)^3 \subseteq \seqcone V$.

Let $P$ be a nondecreasing polynomial witnessing Lemma \ref{lem:zvass-plb}, i.e.,
\tzvass are \lb by $P$.
As $V$ has a path $s \tran t$, it also has a $\Z$-path $s\tran t$.
By Lemma \ref{lem:zvass-plb}, $V$ has a $\Z$-path $s \trans{\sigma} t$ of length at most $P(M)$.
The $\Z$-path factorises into components:
\begin{align} \label{eq:Zpath}
\sigma \ = \ 
\sigma_1;\, u_1; \, \ldots; \, \sigma_{k-1}; \, u_{k-1}; \, \sigma_k.
\end{align}
As in Case 1 of the proof of Lemma \ref{lem:1comp},
let $R$ be a nondecreasing polynomial  such that in every \tvass of size $m$, the length of a covering
path is at most $R(m)$~\cite[Lemma~3.4]{DBLP:journals/tcs/Rackoff78}.
We generalise Lemma \ref{lem:FP} and prove that certain multiplicity of $\Delta'$
may be obtained by executing first a cycle in $V_1$, then a cycle in $V_2$, and so on,
and finally a cycle in $V_k$, so that the total effect of the first $j$ cycles is in $(\Npos)^3$,
for every $j\in\setto k$,
and the lengths of all the cycles are bounded by a polynomial of degree $\OO(k)$:
\begin{lemma} \label{lem:FPgen}
There is an integer cascade $(\Delta'_1, \ldots, \Delta'_k)$ and $\ell\in\Npos$ such that
$\Delta'_1 + \ldots + \Delta'_k = \ell\cdot\Delta'$, 
and for $j\in\setto k$ there are paths
\[
p_j(w+\Delta'_1 + \ldots + \Delta'_{j-1}) \trans {\pi_j}  p_j(w+ \Delta'_1 + \ldots + \Delta'_{j})
\]
in $V_j$
of length $R(M)^{\OO(k)}$.
\end{lemma}
\begin{proof} 
Let $\rho_j$ be a cycle in $V_j$ that visits all states of $V_j$, and
let $\Delta_j\in\Z^3$ be its effect, for $j\in\setto k$.
We have thus $\Z$-paths:
\[
p_j(\Delta_1 + \ldots + \Delta_{j-1}) \trans{\rho_j} p_j(\Delta_1 + \ldots + \Delta_j).
\]
%
%
%
Relying on $\Delta\in(\Npos)^3$,
take a sufficiently large multiplicity $m\in\Npos$ so that 
the $\Z$-paths become paths:
\begin{align} \label{eq:mDe}
p_j(m\cdot \Delta + \Delta_1 + \ldots + \Delta_{j-1}) \trans{\rho_j} p_j(m\cdot \Delta + \Delta_1 + \ldots + \Delta_j).
\end{align}
In particular, the tuple
$(m\cdot\Delta + \Delta_1, \Delta_2, \ldots, \Delta_k)$ becomes a cascade.
Let 
$\widetilde \Delta = m\cdot\Delta + \Delta_1 + \Delta_2 + \ldots + \Delta_k \in \N^3$
be the sum of the cascade.
%
As $\Delta'\in(\Npos)^3$, there is $\ell'\in\Npos$ such that 
$\ell'\cdot\Delta'-\widetilde \Delta \in (\Qpos)^3$,
and hence, by wideness of $(V, s)$, we have
$\ell'\cdot\Delta'-\widetilde \Delta \in \seqcone V$, namely
\[
\ell'\cdot\Delta' - \widetilde \Delta \ = \ s_1  +  \ldots  +  s_k
\]
is the sum of a cascade $(s_1, \ldots, s_k)$, where
\begin{align} \label{eq:presystk} 
\begin{aligned}
s_1 \ =  \ & \ r_{1,1} \cdot e_{1,1} + \ldots + r_{1,n_1} \cdot e_{1,n_1} \\
 & \ \ldots  \\
s_k \ = \ & \ r_{k,1} \cdot e_{k,1} + \ldots + r_{k,n_k} \cdot e_{k,n_k}.
\end{aligned}
\end{align}
for some positive rational coefficients $r_{1,1}, \ldots, r_{1,n_1}, \ldots, r_{k, 1}, \ldots, r_{k,n_k}\in \Qpos$,
where vectors $e_{j,1}, \ldots, e_{j,n_j}\in\Z^3$ are effects of simple cycles in $V_j$, for $j\in\setto k$. Denote by $\rhs_j$ the $j$th right-hand side expression in \eqref{eq:presystk}.
Therefore, the system $\cal S$ consisting of the equation
\begin{align} \label{eq:eqk}
\ell \cdot (\ell'\cdot\Delta'-\widetilde\Delta)  \ = \ \rhs_1  +  \ldots +  \rhs_k
\end{align}
together with the $k$ equations
\begin{align} \label{eq:systk}
\begin{aligned}
\ell_1 \ = \ & \ \rhs_1 \\
\ell_2 \ = \ & \ \rhs_1  +  \rhs_2 \\
& \ldots \\
\ell_k \ = \ & \ \rhs_1  +  \ldots  +  \rhs_k,
\end{aligned}
\end{align}
with unknowns $\ell, \ell_1, \ldots, \ell_k, r_{1,1}, \ldots, r_{k,n_k}$, has a positive integer solution. We rewrite the equation \eqref{eq:eqk} to:
\begin{align} 
\begin{aligned}
\label{eq:syst}
\ell \ell'\cdot \Delta'  \ = \ \ \ell m \cdot \Delta +   (\ell\cdot \Delta_1 + \rhs_1) + \ldots + (\ell \cdot \Delta_k + \rhs_k).
\end{aligned}
\end{align}
%
Irrespectively of the value of $\ell$,
the tuple 
$(\ell m\cdot\Delta + \ell\cdot\Delta_1, \ell\cdot\Delta_2, \ldots, \ell\cdot\Delta_k)$ is still a cascade,
and due to $\ell_j>0$ in \eqref{eq:systk} the tuple
\[
(\widetilde\Delta_1, \ldots, \widetilde\Delta_k) \ := \ (\ell m\cdot\Delta + \ell\cdot\Delta_1 + \rhs_1, \ell\cdot\Delta_2 + \rhs_2, \ldots, \ell\cdot\Delta_k+\rhs_k)
\]
is a cascade as well.
Let $r_j = r_0 + \ell\cdot(\Delta_1 + \ldots + \Delta_j) + \rhs_1 + \ldots + \rhs_j$
denote its $j$th partial sum, for $j\in\setto k$, where $r_0 = \ell m \cdot \Delta$.
%
%
Let $\sigma_{j,i}$ be a simple cycle of effect $e_{j,i}$ in $V_j$.
Let $\sigma_j$
be a $\Z$-path in $V_j$ that starts (and ends) in state $p_j$ and consists of the  
$\ell$-fold concatenation of the cycle $\rho_j$, with attached 
$(r_{j,i})$-fold concatenation of each $\sigma_{j,i}$, for $i\in \setto{n_j}$
(since $\rho_j$ visits all states, this is possible):
\begin{align} \label{eq:sigmajgen}
p_j(r_{j-1}) \trans{\sigma_j} p_j(r_j).
\end{align}
The sum of effect of $\sigma_1, \ldots, \sigma_k$, plus $r_0$, 
yields the right-hand side of \eqref{eq:syst}.
Each of the paths starts and ends in $(\Npos)^3$ but may pass through non-positive points, and therefore it
needs not be a path.
Let $k\in\Npos$ be a multiplicity large enough so that for every $j\in \setto k$, 
the $\Z$-path \eqref{eq:sigmajgen} becomes a path
when lifted by $(k-1) \cdot r_{j-1}$,
i.e., when starting in $p_j(k \cdot r_{j-1})$, 
and also becomes a path when lifted by $(k-1)\cdot r_j$, i.e., when
ending in 
$p_j(k \cdot r_j)$.
In this case, the $k$-fold concatenation of each $\sigma_j$ is also a path:
\[
p_j(k\cdot r_{j-1}) \trans{(\sigma_j)^k} p_j(k \cdot r_j),
\]
since all points visited in the inner iterations of $\sigma_j$ are bounded 
from both sides by corresponding points visited in the first and the last iteration of $\sigma_j$.
Therefore, the lifting of $(\sigma_j)^k$ by the source vector $w$ is also a path:
\begin{align} \label{eq:thepathgen}
p_j(w+k\cdot r_{j-1}) \trans{(\sigma_j)^k} p_j(w+k \cdot r_j).
\end{align}
Relying on \eqref{eq:thepathgen},
we define the cascade
\[
(\Delta'_1, \ldots, \Delta'_k) \ := \ (k \cdot \widetilde\Delta_1, \ldots, k \cdot \widetilde\Delta_k),
\]
and cycles $\pi_1, \ldots, \pi_k$:
the cycle $\pi_1$ is  $(\sigma_1)^k$ precomposed with the $(k\ell m)$-fold iteration
of $\pi$,
\[
p_1(w) \trans{\pi^{k\ell m}} p_1(w + k\cdot r_0) \trans{(\sigma_1)^k} p_1(w + \Delta'_1),
\]
and for $j\in\setfromto 2 k$, the cycle $\pi_j$ is $(\sigma_j)^k$,
\[
p_j(w + \Delta'_1 + \ldots + \Delta'_{j-1}) \trans{(\sigma_j)^k}
p_j(w + \Delta'_1 + \ldots + \Delta'_{j}),
\]
as required.

The estimations of the length of the paths $\pi_j$ are similar as in the proof of Lemma \ref{lem:FP},
so we focus only on new aspects.
The number of different effects of cycles in each component is at most $(2M+1)^3\leq \OO(M^3)$, 
and therefore
the number of unknowns $r_{j,i}$ in $\cal S$
is at most $k\cdot (2M+1)^3 \leq \OO(M^4)$, and consequently so is the 
total number of unknowns in $\cal S$.
The norm of a solution of the system $\cal S$ can now be bounded, due to Lemma \ref{lem:taming},
by $D = R(M)^{\OO(k)}$.
In consequence, we get the same bound $R(M)^{\OO(k)}$ on lengths of paths $\pi_j$.
This completes
the proof of Lemma  \ref{lem:FPgen}.
\end{proof}

We now use Lemma \ref{lem:FPgen} to complete the proof of Lemma \ref{lem:dw}.
Note that $\ell$ in Lemma \ref{lem:FPgen} is necessarily also bounded by $R(M)^{\OO(k)}$, and that
for every $m\in\Npos$ the $m$-fold iteration of the cycle $\pi_j$ is also a path:
\begin{align}\label{eq:mfold}
p_j(w+m\cdot(\Delta'_1 + \ldots + \Delta'_{j-1})) \trans {(\pi_j)^m} p_j(w+ m\cdot(\Delta'_1 + \ldots + \Delta'_{j})).
\end{align}
We pick an $m\in \Npos$ and 
build a path $\rho$
by interleaving the $\Z$-path \eqref{eq:Zpath} with
$m$-fold iterations of the cycles $\pi_1, \ldots, \pi_k$ of Lemma \ref{lem:FPgen}:
\[
\rho \ = \ (\pi_1)^m;\, \sigma_1;\ u_1; \ \ldots; \ \ (\pi_{k-1})^m;\, \sigma_{k-1};\ u_{k-1}; \ (\pi_k)^m;\, \sigma_k;  
\]
As the effect of $(\pi_j)^m$ is $m\cdot \Delta'_j$, and the effect of $\sigma$ is $w'-w$,
we have:
\[
p_1(w) \trans{\rho} p'_k(w'+m\ell\cdot \Delta').
\]
We choose $m$ sufficiently large to enforce that each of $\Z$-paths $\sigma_j$
becomes a path, and hence the whole $\rho$ is a path as well.
It is enough to take 
 $m = M\cdot P(M)$, which
makes the length of $\rho$ bounded by $P(M) \cdot R(M)^{\OO(k)}$.
%
%
Finally, we concatenate $\rho$ with the $m\ell$-fold iteration of the path $\pi'$,
\[
p'_k(w' + m\ell\cdot \Delta') \trans{(\pi')^{m\ell}} p'_k(w'),
\]
to get the required path $\rho;\, (\pi')^{m\ell}$
from $p_1(w)$ to $p'_k(w')$
of length bounded by $M\cdot P(M) \cdot R(M)^{\OO(k)} \leq
(M\cdot P(M)\cdot R(M)^{\OO(1)})^k \leq M^{\OO(k)}$.
This completes the proof of Lemma \ref{lem:dw}. 
\end{appendixproof}

For stating and proving Lemma \ref{lem:len-eq} we need a variant of sequential \tvass:
a \emph{good-for-induction} $k$-component \tvass $V=\ktvass$ is defined exactly like $k$-component
sequential \tvass, except that the first component $V_1$ is an arbitrary \geomvass,
not necessarily being strongly connected.
\begin{lemma} \label{lem:len-eq}
There is a nondecreasing polynomial $R$ such that 
every non-easy $k$-component \tvass $(V, s, t)$ is length-equivalent to a finite set $S$ of good-for-induction
$k$-component \tvass
of size at most $R(\size(V,s,t))$, namely
$\Len V s t = \bigcup_{(V', s', t')\in S} \Len {V'} {s'} {t'}$.
\end{lemma}

\begin{appendixproof}[Proof of Lemma \ref{lem:len-eq}]
Consider a non-easy $k$-component \tvass $V = \ktvass$, together with
source and target configurations $s = p_1(w)$ and $t=p'_k(w')$.
If $V_1$  is a \geomvass, there is nothing to prove as $V$ is good-for-induction.
If $\rev{(V_k)}$ is a \geomvass then we are done too, as $\rev V$ is good for induction and 
$\Len V s t = \Len{\rev V} t s$.
Therefore we assume from now on that $V_1$ and $\rev{(V_k)}$ are of geometric dimension 3. 
In consequence, by Claim \ref{claim:seqcone}, all of $\cone {V_1}, \seqcone V$, $\cone{\rev{(V_k)}}$, $\seqcone{\rev V}$ are
3-dimensional open cones.
We distinguish two cases, and hence the polynomial $R$ is the sum of polynomials claimed
in the respective cases.

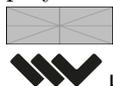
\para{Case I: $(V, s, t)$ is non-diagonal}
We may assume \mywlog that $V$ is non-forward-diagonal (otherwise replace $V$ by $\rev V$),
and therefore $V_1$ is so.
Exactly as in Case 3 of the proof of Lemma \ref{lem:1comp}, we transform
$(V_1,s)$ into
three \geomvass $(\essdvass V_1, s_1), (\essdvass V_2, s_2), (\essdvass V_3, s_3)$.
In each $(\essdvass V_i, s_i)$, we replace the target state $p'_1$ by
$\pair {(p'_1)} b$, for an arbitrarily chosen value $b\in\setfromto 0 B$ of the bounded coordinate,
and modify accordingly the first bridge transition $u_1 = (p'_1, \delta_1, p_2)$ to 
$\essdvass u_{i, b} = (\pair{(p'_1)} b, \delta_1, p_2)$.
This yields a set $S$ of $3(B+1)$ good-for-induction $k$-component \tvass  
\[
S \ = \ \setof{\big(\ktvassmod, s_i, t \big)}{i\in\setto 3, \ b\in \setfromto 0 B},
\]
which is length-equivalent to $(V, s, t)$, as required. 
The size of each of these \tvass is
at most $R(M)$, as in Claim \ref{claim:nondiagsize} in Case 3 of the proof of Lemma \ref{lem:1comp}.

\para{Case II: $(V, s, t)$ is non-wide}
We proceed similarly to Case 2 of the proof of Lemma \ref{lem:1comp}, and transform
$(V_1,s)$ into a \geomvass $(\essdvass V, \essdvass s)$, defined as a \mytrim {$a$} {$B$} of $V_1$
for some vector $a\in\Z^3$ and $B\in\N$.
To this aim we need an analog of Claim \ref{claim:ax} (Claim \ref{claim:axgen} below).
\begin{claim}\label{claim:empty_intersection}
$C_1 :=  \cone{V_1}$ and $S := \seqcone{\rev V}$ are disjoint.
\end{claim}

\begin{proof}
Let $S'=\seqcone {\rev{(V')}}$, where $V'=\ktvasstwo$ is $V$ without the first component and the first bridge.
By definition, 
$
S = (S' + \cone{\rev{(V_1)}})\cap(\Qpos)^3,
$
but since $\cone{\rev{(V_1)}} = -C_1$, we get
\[
S = (S' - C_1)\cap(\Qpos)^3.
\]
By Claim \ref{claim:seqcone} the set $S'$ is an open cone, and hence
$S' - C_1$ is also so, as Minkowski sum preserves such cones.
By definition $S' - C_1$ contains, for every vector $v\in C_1$, 
a vector $\varepsilon-v$ for some vector $\varepsilon\in S'$ of arbitrarily small norm $(*)$.

Towards a contradiction, suppose $C_1 \cap S$ is nonempty.
Therefore, $S$, and hence also $S'-C_1$ contains some vector $v\in C_1$.
Being an open cone, it also contains $v - \varepsilon$ for every vector $\varepsilon$ of sufficiently small norm
$(**)$.
The two properties $(*)$ and $(**)$,
\[
(*) \ \ \ \prettyforall {v\in C_1} \prettyforall {N\in\Q} \prettyexists {\varepsilon, \norm(\varepsilon)<N} 
(\varepsilon - v \in S'-C_1)
\qquad
(**) \ \ \ \prettyexists {v\in C_1} \prettyexists{N\in\Q} \prettyforall {\varepsilon, \norm(\varepsilon)<N} (v- \varepsilon  \in S'-C_1),
\]
imply that $S'-C_1$ contains both $v-\varepsilon$ and $\varepsilon-v$, for some
$v\in C_1$ and $\varepsilon\in\Q^3$,
and hence $S'-C_1$
includes a line.
As $S-C_1$ is an open cone, we deduce $S'-C_1=\Q^3$ and  $S=(\Qpos)^3$, which means that 
$V$ is wide, a contradiction.
%
%
\end{proof}
%
%
\begin{claim} \label{claim:onefacet}
Let $C\subseteq \Q^3$, $C'\subseteq (\Qpos)^3$ be 3-dimensional disjoint open cones, and $D\in\Qpos$. 
All points whose distance to both cones is at most $D$,
are at distance at most $3D$ to one of facet planes of $C$.
\end{claim}
\begin{proof}
We give a geometric argument.
Let $S$ be any plane \emph{separating} $C$ and $C'$, namely the two cones are on the opposite sides of $S$.
Since $C'\subseteq(\Qpos)^3$ is included the positive quadrant,
the plane $S$ may be chosen to be adjacent to $C$, namely to satisfy one of the following conditions:

\begin{itemize}
\item[(1)] $S$ includes a facet $F$ of $C$, or
\item[(2)] $S$ includes an edge of $C$ (adjacent to two facets $F_1$ and $F_2$).
\end{itemize}

\noindent
Consider an arbitrary point $x\in\Q^3$ such that $d(x, C)\leq D$ and $d(x, C')\leq D$.
Therefore $d(x,S) \leq D$, as $C$ and $C'$ are on opposite sides of $S$.
Let $x'\in S$ be the point in $S$ which is the closest to $x$.
In case (1), we have $d(x,S)\leq D\leq 3D$, i..e, $x$ is at distance at most $D$ to the facet plane $S$.
In case (2), let $H_1, H_2$ be the planes including $F_1$, $F_2$, respectively.
Since $d(x, C)\leq D$ and $d(x,S) \leq D$, we deduce that $d(x,H_1)\leq D$ or $d(x,H_2)\leq D$.
\Wlog we assume that the angle between $H_1$ and $S$ is at most as large as the angle between $H_2$ and $S$,
and aim at showing $d(x, H_1) \leq 3D$.
If $d(x, H_1)\leq D$, we are done.
Otherwise $d(x, H_2) \leq D$, and hence by the triangle inequality we get
$d(x', H_2) \leq d(x', x) + d(x, H_2) \leq 2D$.
Since the angle between $H_1$ and $S$ is not larger than the angle between $H_2$ and $S$, 
we deduce $d(x', H_1) \leq d(x', H_2) \leq 2D$, which implies
$d(x, H_1)\leq d(x, x') + d(x', H_1) \leq 3D$, as required.
%
%
%
%
As $x$ was chosen arbitrarily, this completes the proof.
\end{proof}
Let $B:=9\cdot D^2 \cdot P(M)^2$, where $P$ comes from Lemma \ref{lem:not_far_from_cone} and
$D\leq \OO(M^2)$ from Claim \ref{claim:a}.
\begin{claim} \label{claim:axgen}
There is a vector $a\in\Z^3$ of $\norm(a)\leq D$ such that
all configurations $q(x)$ in $V_1$ appearing on a path $s\tran t$ in $V$ satisfy
$-B \leq \innprod a x \leq B$.
\end{claim}
\begin{proof}
Consider a path $s \trans{\pi} t$ and let $\pi_1$ be its prefix in $V_1$.
By Lemma~\ref{lem:not_far_from_cone}, all the vectors appearing in $\pi_1$ are not further than $P(M)$ from 
$S=\seqcone{\rev V}$, but also not further than $P(M)$ from $C_1 = \cone{V_1}$.
By disjointness of $C_1$ and $S$, due to Claim \ref{claim:empty_intersection}, 
and by Claim \ref{claim:onefacet},
there is a facet $F$ of $C$ such that
all the vectors $x$ appearing in $\pi_1$ are at distance at most  $3\cdot P(M)$ to
the hyperplane $H$ including $F$.
Due to Claim~\ref{claim:a}, $H=\setof{y}{\innprod {a} y = 0}$ for some
vector $a \in \Z^3$ of $\norm(a) \leq D=\OO(M^2)$.

In order to bound the value of $\innprod {a} x$ for an arbitrarily vector $x$ appearing in $\pi_1$,
split $x$ into $x = x_1 + x_2$, where $x_1$ is orthogonal to $H_i$, i.e., $x_1 = \ell \cdot a$ for some $\ell\in\Q$.
\Wlog assume that the length of $a$ is at least 1.
The length of $x$  is $\innprod {x_1}{x_1}\leq 9\cdot P(M)^2$, and hence
$x_1 = \ell \cdot a$ for some $\ell$ satisfying $\absv \ell \leq 9\cdot P(M)^2$.
In consequence,  $\innprod {a} x = \innprod {a} {x_1} = \ell \cdot \innprod {a} {a}$,
and hence $\absv {\innprod {a} x} \leq \ell \cdot \norm(a)^2 \leq B= 9\cdot D^2 \cdot P(M)^2$, as required.
%
\end{proof}
We complete the proof of Case II as in Case 2 of the proof of Lemma \ref{lem:1comp}.
We replace the first component $V_1$ by the \geomvass $\essdvass V$, as defined there,
of size $E = \OO(M\cdot B)$ (as stated in Claim \ref{claim:sizeVV}),
and the source configuration $s$ by $\essdvass s$.
We also replace the first bridge transition $u_1 = (p'_1, \delta_1, p_2)$ by
$\essdvass u_{b} = (\pair{(p'_1)} b, \delta_1, p_2)$, for any $b\in\setfromto{-B} B$.
This yields a set $S$ of $2B+1$ good-for-induction $k$-component \tvass
\[
S \ = \ \setof{\big(\ktvassmodmod, \essdvass s_a, t \big)}{a\in A, \ b\in \setfromto {-B} B},
\]
which is length-equivalent to $(V, s, t)$, as required. 
The size of each of these \tvass is at most $R(M) = E+M \leq \OO(M \cdot B)$.
\end{appendixproof}


\begin{proof}[Proof of Lemma \ref{lem:ne}]
Relying on Lemma \ref{lem:len-eq}, assume \mywlog that
$(V, s, t)$ is a 
good-for-induction $k$-component \tvass of size
$\size(V, s, t) = R(M)$, where $R$ is a polynomial of Lemma \ref{lem:len-eq}.
Let $V = \ktvass$, $u_i = (p'_i, \delta_i, p_{i+1})$,
$s=p_1(w)$ in $V_1$ and $t=p'_k(w')$ in $V_k$.
Let $F$ be a nondecreasing polynomial witnessing 
Lemma \ref{lem:geom-sandwich}, and
let $f(x) = 2R(x) \cdot F(R(x))$. 
Let $G$ be a nondecreasing polynomial witnessing Lemma  \ref{lem:geom-plb}, and
let $g(x) = G(2x^2) + x$.

Let $B := h(f(M))$, namely the length-bound for $(k-1)$-component \vass of size $f(M)$.
Assuming $s\trans{\pi} t$ in $ V$, we aim at proving that there is such a path of length at most
$g(h(f(h(f(M)))))$.
Decompose the path $s\tran t$ into
\begin{align} \label{eq:pathV}
s \trans{\pi_1} p'_1(v') \trans{u_1} p_2(v) \trans{\pi'} t,
\end{align}
where $p_2(v)$ is the first configuration of $V_2$ appearing on $\pi$.
As $V_1$ is a \geomvass, so is the \emph{lollypop} \tvass $V'_1$ 
obtained by adding to $V_1$ the first bridge transition $u_1$
(indeed, adding a bridge transition does not create any new cycles).
As $\size(V,s,t) \leq R(M)$ and $F$ is the polynomial witnessing Lemma~\ref{lem:geom-sandwich},
by Lemma~\ref{lem:geom-sandwich} we get:
\begin{claim}
The set $\reach_{p_2}(V,s)$ is  \kanapka {$F(R(M))$} {$B$}.
\end{claim}

Thus $v \in L = a + P^*$, where $L$ is a linear set ($*$) or a $B$-approximation thereof ($**$).
In both cases  $v = a + r$, for $r\in P^*$.
We construct a $(k-1)$-component \tvass $(V', s', t)$ as follows.
$V'$ is obtained from $\ktvasstwo$ by adding,
for every period vector $r\in P$, a self-looping transition
$(p, r, p)$ to $ V_2$.
As the source configuration we take $s' = p_2(a)$, and keep $t$ as the target configuration.

As $R$ is nondecreasing, $B = h(f(M))$ and $f \geq R$ we have $M \leq R(M) \leq B$.
As $(V',s',t')$ is obtained from of $(V,s,t)$ by removing from $V$ the part $V_1$,
adding transitions of total norm at most equal $\norm(P)$ and setting $s' = p_2(a)$ 
therefore we estimate  $M' = \size(V', s', t)$ as 
$
M' \ \leq \ R(M) + \norm(a) + \norm(P).
$
Recall that in the case when $L$ is a linear set ($*$) we have $\norm(a) \leq B \cdot A$, $\norm(P) \leq A$,
while in the case when $L$ is a $B$-approximation ($**$) we have $\norm(a), \norm(P) \leq A$,
in our case $A = F(R(M))$.
Thus the estimations
look as follows:
\begin{align*}
(*) \quad M' & \leq R(M) + (B+1) \cdot F(R(M)) \leq  2\cdot R(B) \cdot F(R(B)) = f(B)\\
(**) \quad M' & \leq R(M) + 2 \cdot F(R(M)) \leq 2\cdot R(M) \cdot F(R(M)) = f(M),
\end{align*}
as $f$ was defined exactly as $f(x) = 2R(x) \cdot F(R(x))$.
Note that the latter bound is dominated by the former one, thus in both cases $M' \leq f(B) = f(h(f(M)))$.

\smallskip

There is a path $s' \tran t$ in $V'$, namely the concatenation of  
a path $p_2(a) \tran p_2(v)$ using the 
just added self-looping transition in the state $p$, with the suffix $\pi'$ of \eqref{eq:pathV}:
\[
p_2(a) \tran p_2(v) \trans{\pi'} t.
\]
By induction assumption, there is a path $s' \tran t$ in $V'$ of length at most
\begin{align*}
(*) \quad & \ell \ = \ h(M') \ \leq \ h(f(B)) \ = \ h(f(h(f(M))))\\
(**) \quad &  \ell \ = \ h(M') \ \leq \ h(f(M)).
%
\end{align*}
\Wlog we may assume that the path executes all just added self-looping transitions in the beginning,
and therefore it splits into:
\[
p_2(a) \trans{\rho''} p_2(v'') \trans{\rho'} t
\]
such that the suffix $\rho'$ is actually a path in $\ktvasstwo$.
Since the length of the prefix $\rho''$ is at most $\ell$, 
we bound $\norm(v'') \leq M \cdot \ell$.
If $L$ is a linear set, the following claim is obvious since $L\subseteq\reach_{p_2}(V,s)$.
On the other hand, the claim requires a proof in case when $L$ is a $B$-approximation of a linear set:
\begin{claim}[$**$]
There is a path $s \tran p_2(v'')$ in the lollypop \tvass $V'_1$.
\end{claim}
Indeed,  
since the length $\ell$ of the path $s'\tran t$ in $V'$ is at most $h(f(M))$
and $B= h(f(M))$, we deduce that $v'' \in a + P^{\leq B}$, and therefore $s\tran p_2(v'')$ in $ V'_1$.

\smallskip

By Lemma \ref{lem:geom-plb}, there is a path $s\trans{\rho} p_2(v'')$ in $V'_1$ of length at most
$G(M + \norm(v'')) \leq G(M\cdot(\ell+1)) \leq G(2\ell^2)$.
Concatenating the two paths $\rho$ and $\rho'$ we get a path
\[
s\trans{\rho} p_2(v'') \trans{\rho'} t
\]
in $V$, of length at most 
$
G(2\ell^2) + \ell = g(\ell) \leq g(h(f(h(f(M))))).
$
Taking $H = f+g$, the sum of polynomials $f$ and $g$, we get the bound
$H(h(H(h(H(M)))))$,
as required.
\end{proof}


\section{Future research}\label{sec:future}
Below we list a few research questions, which we find interesting and
particularly promising directions after our contribution.

\para{Exact complexity for $3$-VASS}
We have shown that shortest paths in binary $3$-VASS are of at most triply-exponential length.
It is tempting to conjecture that actually the upper bound for the length of the paths is shorter,
at most doubly-exponential. We conjecture that it is possible with techniques similar to the developed ones,
but with more focus on polynomials growing linearly with respect to norms of source and target.
We leave proving this conjecture to the future research.

\para{Example of a $3$-VASS with doubly-exponential path}
We have shown that shortest paths in binary $3$-VASS are of at most triple-exponential length.
However, currently we still do not know any example in which even a path of doubly-exponential length is needed,
it might be that paths of exponential length are sufficient leading to \pspace-completeness for binary $3$-VASS.
It would be very interesting to find an example of a binary $3$-VASS with shortest path between two configurations
being doubly exponential. An example of binary $4$-VASS of doubly-exponential shortest path is known (see Section 5 in~\cite{DBLP:conf/concur/Czerwinski0LLM20}). Maybe some modification of this $4$-VASS would allow to design a $3$-VASS with similar properties.

\para{Reachability for $d$-VASS with $d \geq 4$}
It is a natural question whether our techniques extend to higher dimensions.
The answer is: possibly yes, but we would need a few other structural results for $3$-VASS
to make a similar approach to $4$-VASS possible. In the proof of Lemma~\ref{lem:main} we do not only
use $2$-VASS reachability as a black box, but we use a deep understanding of the reachability relation in $2$-VASS
from~\cite{DBLP:conf/focs/0001CMOSW24}. Probably a similar understanding of the reachability relation for $3$-VASS would be needed
to advance understanding of $4$-VASS along our lines. 

In general it is very interesting to determine the complexity of the reachability problem for $d$-VASS.
We have excluded that for each $d \geq 3$ the problem is $\F_d$-completely, but it is still possible that
the problem is $\F_{d-C}$-complete for some constant $C \in \N$ and $d$ big enough.
Recall that in~\cite{DBLP:conf/fsttcs/CzerwinskiJ0LO23}
it was shown that the reachability problem for $(2d+4)$-VASS is $\F_d$-hard for any $d \geq 3$ and this
is the best currently known lower bound for arbitrary dimension.
Therefore the other natural possibility is that the reachability problem for $(2d+C)$-VASS is $\F_d$-complete for some
constant $C \in \N$. 

\para{Applications of the approximation technique}
Another natural research direction is to search for other applications of the technique of approximating the reachability sets,
which allows to lower the complexity down, below the size of the reachability set.
One particular case, which seems to be prone to such techniques is the $2$-VASS with some number of $\Z$-counters, namely counters, which can take values below zero.
The best complexity lower bound for the reachability problem in this model is \pspace-hardness inherited from~\cite{BlondinFGHM15},
while the best upper bound is Ackermann membership inherited from VASS reachability~\cite{LS19}.
The reachability sets for that systems are not necessarily semilinear.
This disqualifies most of the techniques relying on the semilinearity of reachability sets, but our techniques
seem to be promising for that model.


%
\paragraph*{Acknowledgements}
We thank Radosław Piórkowski for many inspiring discussions about cones in VASS a few years ago
and Filip Mazowiecki for the discussions about the reachability problem for \tvass.
We thank also Henry Sinclair-Banks for helpful discussions about \cite[Thm 4.16]{DBLP:conf/focs/0001CMOSW24}.
We thank Qizhe Yang and Yangluo Zheng for the discussions on $3$-VASS, which are geometrically $2$-dimensional.
Finally, we thank anonymous reviewers for helpful comments.


\bibliographystyle{plain}
\bibliography{citat}

\newpage

\appendix

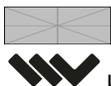
\end{document}